\documentclass[12pt, english, a4paper]{amsart}

\usepackage{graphicx}

\usepackage{amsmath,amssymb}
\usepackage{bm}
\usepackage{euscript,color}
\usepackage{hyperref}
\usepackage{babel}
\usepackage{amsthm}
\usepackage{amsfonts}
\usepackage{latexsym}
\usepackage{fancyhdr}
\usepackage{enumerate}
\usepackage{multirow}

\usepackage{array}

\usepackage{tikz}\usepackage{pgfplots}
\usetikzlibrary{intersections,calc}
\usetikzlibrary{arrows,decorations.markings}
\pgfplotsset{compat=1.12}
\tikzstyle arrowstyle=[scale=12pt]

\usepackage[left=2.5cm, top=3cm, bottom=3cm, right=2.5cm]{geometry}

\newcommand{\R}{\mathbb{R}}

\makeatletter
\renewcommand{\section}{\@startsection%
{section}
{1}
{0mm}
{1.5\bigskipamount}
{0.5\bigskipamount}
{\centering\normalsize\sc}}

\renewcommand{\paragraph}{\@startsection%
{paragraph}
{4}
{0mm}
{\bigskipamount}
{-1.25ex}
{\normalsize\sl}}

\def\provedboxcontents#1{$\square$}

\makeatother

\newtheoremstyle{thm}{}{}{\slshape}{}{\scshape}{.}{0.5em}{}

\newtheoremstyle{def}{}{}{}{}{\scshape}{.}{0.5em}{}

\newtheoremstyle{rmk}{}{}{}{}{\scshape}{.}{0.5em}{}

\newtheoremstyle{claim}{}{}{}{}{\slshape}{.}{0.5em}{}

\theoremstyle{thm}
\newtheorem{newstatement}{newstatement}

\newtheorem{theorem}[newstatement]{Theorem}

\newtheorem*{conjecture*}{Conjecture}

\theoremstyle{def}

\theoremstyle{rmk}
\newtheorem{remark}[newstatement]{Remark}

\theoremstyle{claim}

\expandafter\let\expandafter\oldproof\csname\string\proof\endcsname
\let\oldendproof\endproof
\renewenvironment{proof}[1][\proofname]{%
  \oldproof[\slshape #1]%
}{\oldendproof}

\let\geq\geqslant
\let\leq\leqslant

\let\epsilon\varepsilon

\renewcommand{\emph}[1]{{\slshape #1}}
\renewcommand{\em}{\sl}

\title[A geometrical approach]{Interpreting the outcomes of research assessments: a geometrical approach}

\author[B. Cappelletti-Montano]{Beniamino Cappelletti-Montano}
  \email{b.cappellettimontano@unica.it}

\author[S. Columbu]{Silvia Columbu}
 \email{silvia.columbu@unica.it}

\author[S. Montaldo]{Stefano Montaldo}
 \email{montaldo@unica.it}

\author[M. Musio]{Monica Musio}
 \email{mmusio@unica.it}

 \address{Dipartimento di Matematica e Informatica, Universit\`a degli Studi di
 Ca\-gli\-ari, Via Ospedale 72, 09124 Cagliari, Italy}

\date{}

\thanks{Research partially supported by the project STAGE funded by Fondazione di Sardegna. B.C.M. and S.M. are members of INDAM (Istituto Nazionale di Alta Matematica)}


\begin{document}
\begin{abstract}
Research evaluations and comparison of the assessments of academic institutions (scientific areas, departments, universities etc.) are among the major issues in recent years in higher education systems. One method, followed by some national evaluation agencies, is to assess the research quality by the evaluation of a limited number of publications in a way that each publication is rated among $n$ classes. This method produces, for each institution, a distribution of the publications in the $n$ classes.  In this paper we introduce a natural geometric way to compare these assessments by introducing an ad hoc distance from the performance of an institution to the best possible achievable assessment.  Moreover, to avoid the methodological error of comparing non-homogeneous institutions, we introduce a {\em geometric score} based on such a distance. The latter represents  the probability that an ideal institution, with the same configuration as the one under evaluation, performs worst. We apply our method, based on the geometric score, to rank, in two specific scientific areas, the Italian universities 
using the results of the evaluation exercise  VQR 2011-2014. 

\end{abstract}

\maketitle
\section{Introduction}

Assessing the quality and/or impact of research of a given institution (university or research center) and assessing the corresponding improvement over time is one of the most difficult tasks in modern quality assurance systems. 
On the other hand, the growing development of university rankings shows that research performance is perceived to be related to the reputation of universities and, in some countries, the allocation of funds in the higher education system is linked to the research performance of institutions. 

Several methods for assessing research performances of universities have been adopted. The most controversial is the use of bibliometric indicators, such as the number of publications, total citations and/or journals impact factor. 
In this regard, it should be mentioned that various scientific associations have signed the so-called DORA declaration \cite{DORA}, which states the contrariety to the automatic use of bibliometrics in order to allocate funding for research and /or evaluate the careers of individual researchers.

Another one, much more expensive,  is the method followed by some national evaluation agencies, which assess the quality of a limited number of publications for each university, using peer-review, informed peer-review, bibliometrics methods or both, according to the scientific area. The final aim of these methods is to rate each publication among $n$ classes of quality (usually, $n$ is taken to be $4$ or $5$). For instance, in the last call of UK Research Excellence Framework (REF 2021) and Italian Quality Research Evaluation (VQR 2015-2019)  are suggested, respectively, 4 (1*, 2*, 3*, 4*, apart of unclassified) and 5 (A - excellent and extremely relevant, B - excellent, C - standard, D - sufficiently relevant, E - poor or not acceptable) classes.

There is a lively debate on the procedures, criteria and methods used in these exercises (see for instance \cite{Abramo:2016}, \cite{Baccini}, \cite{Demetrescu}, \cite{Franceschini}, \cite{taylor1995}, \cite{varin2016}). However, in this paper we do not want to go into the substance of the methodology and effectiveness of such evaluation exercises. Our aim is to consider  the more subtle question of how researchers, policy makers, citizens could / should interpret the output data. Namely, any university and any department receive an evaluation in terms of the percentage  of the submitted publications evaluated in each of the classes stated in the call. Here there are two main issues. First, the assessment, given in this way, appears to be \emph{absolute} and there is the dangerous temptation to compare directly the performances of two universities and/or departments whose composition could be very different from each other.
For instance, since the number of required publications is the same for any scientific area, it is likely that a Department of Physics will obtain, on average, better evaluations than, for instance, a Department of Law, just because of the very high average number of articles per researcher in Physics. 
Second, even assuming to compare two homogeneous aggregations (for instance the Departments of Physics of two universities, or the same Department of Physics along two or more editions of the evaluation exercise), it could not be a trivial task to understand if one assessment is better or not than the other. Of course, if the two assessments compared are those represented in the following table
\medskip
\begin{center}
\begin{tabular}{|c|c|c|c|c|c|}
\hline
Institution         &  A        &  B & C & D & E      \\
\hline
Department 1            & 100\%       & 0\% & 0\% & 0\% & 0\%          \\
Department 2         & 0\%     &  100\%    & 0\%   & 0\%   & 0\%          \\
\hline
\end{tabular}
\end{center}
\medskip
then it is easy to say that the assessment of Department 1 is better than that one of Department 2. But if the two assessments are, for instance, the following ones
\medskip
\begin{center}
\begin{tabular}{|c|c|c|c|c|c|}
\hline
Institution         &  A        &  B & C & D & E      \\
\hline
Department 1            & 20\%       & 20\% & 20\% & 20\% & 20\%          \\
Department 2         & 15\%     &  25\%    & 21\%   & 21\%   & 18\%          \\
\hline
\end{tabular}
\end{center}
\medskip
then it is not straightforward  to decide which department got the best evaluation. A further complication is the usual presence of departments composed by different scientific areas, which could lead the political decision-maker to misleading analyses of the output data arising from the evaluation exercise: taking inspiration from some real cases in Italian universities, how to compare the performances of a ``Department of Mathematics, Computer Science and Economics'' and a ``Department of Mathematics and Geosciences''? And how to compare the performances of all departments in the same university?

In this paper we try to address all these questions. Of course there is not just \emph{one} answer, since  several methods for interpreting data can be defined. However, we show that the above are essentially \emph{geometric} questions and that there is a natural geometric way to treat this topic.

Our first observation is that the outcomes in an evaluation exercise (for instance the data in the above tables) can be geometrically represented as points of the standard simplex $\Delta_n$, where $n+1$ is the number of attributes involved in the call. Then the overall assessment of a department can be ``measured'' as the ``distance'' between the point $P_0$ in  $\Delta_n$, representing the evaluation of the department, and the vertex $P_1=(1,0, \ldots, 0)$ of the simplex, which corresponds to the best possible assessment. We define such distance $\delta(P_0)$ as the length of a natural path joining $P_0$ to $P_1$, corresponding to improving the assessment of the department in the slowest smooth possible way (see Section \ref{natural-path} for details).  By the application of the beautiful geometric properties of the simplex, we find an iterative method to determine this path, obtaining a general formula for $\delta(P_0)$.

This procedure permits to associate a real number to any assessment expressed by ordinal variables. However, $\delta$ can not be used directly for comparing different departments, unless they are reasonably homogeneous. At least they should have the same size and the same internal structure in terms of research areas. In fact,  random variations are larger for small samples, so that  evaluation results tend to be ``funnel-shaped": for mega-universities it is difficult to deviate much from the average (narrow part of the funnel) while among the small ones (large part) it is frequent to see exceptional results, both for positive and negative performances.  Furthermore, each scientific area has its own peculiarities, citational trends and editorial practices, making meaningless to compare any two different scientific areas. Depending on the availability of data and on the aims of the evaluation, one can consider other possible homogeneity criteria, such as teaching duties of professors, salary, age, gender, and so on. Let $\mathcal{C}$ denote the set of all the homogeneity criteria chosen.   We can consider the set of all (ideal) departments, whose members are randomly selected from all the universities participating to the call, so that they have the same size as a given Department $D$ and satisfy the same homogeneity criteria $\mathcal{C}$ when compared to $D$. Then we can define the \emph{geometric score} of  $D$ as the proportion ${\mathcal S}_{\mathcal{C}}(D)$ of such ideal departments $D'$ for which $\delta(D')> \delta(D)$.   In other words, ${\mathcal S}_{\mathcal{C}}(D)$ represents the probability that an ideal department $D'$, with the same configuration as $D$ (hence comparable with $D$), performs worst than $D$. As we shall show in the article, it has also a nice geometric interpretation.

In this way one compares any department - and, more in general, any ``aggregate'', including a whole university itself - with its similars  only (in fact with all their possible similars).  This procedure avoids the  methodological error, very frequent in several research assessment exercises as well as in many university rankings, of comparing universities, departments and scientific areas which, in principle, cannot be compared directly.

One issue related to the geometric score is its computability. Even in the case of a small department,  the cardinality of the set of all ideal departments is a very large number, and the exact calculation of the geometric score is not practicable. However,  we can approximate the geometric score using Monte Carlo techniques that guarantee the almost sure convergence of the estimate to the geometric score. In the last part of the paper we use a simple algorithm  for the calculation of the geometric score for some aggregates of the Italian VQR 2011-2014. Namely we deal with the areas  of  ``Mathematics'' and of ``Statistics and Mathematical Methods for Decisions'', which are composed, respectively, of more than 2000 and 1000 professors in Italy. We show an easily implementable way to compute good approximations of the geometric score, and, interestingly, we find that the geometric score ranking is very different from the official ANVUR ranking which is still in use to allocate conspicuous public fundings to Italian universities.

\section{Preliminary notions: the geometry of the $n$-simplex}
Let  $P_{1},\ldots ,P_{n+1}\in \mathbb{R} ^{n+1}$ be $n+1$ points of $\mathbb{R} ^{n+1}$ which are affinely independent, i.e. the vectors  $P_{2}-P_{1}, \ldots,  P_{n+1} - P_{1}$ are linearly independent. Then, the \emph{$n$-simplex} determined by $P_{1},\ldots ,P_{n+1}$ is the subset of $\mathbb{R}^{n+1}$ given by
\begin{equation*}
\Delta_{P_{1},\ldots,P_{n+1}} := \left\{x_{1}P_{1} + \cdots + x_{n+1}P_{n+1} \colon  x_i \geq 0 \hbox{ for all $i=1,\ldots,n+1$ and } \sum_{i=1}^{n+1}x_{i} = 1 \right\}.
\end{equation*}
The convex hull of any non-empty subset of cardinality $m+1$ of  $\left\{P_{1},\ldots ,P_{n+1}\right\}$ is, in turn, a simplex, called $m$-\emph{face}. In particular  $0$-faces, i.e. the defining points $P_{1}, \ldots , P_{n+1}$ of the simplex, are called the \emph{vertices},  $1$-faces are called the \emph{edges}, and $n$-faces are called the \emph{facets}.

If one takes the points $P_{1}=(1,0,\ldots,0), P_{2}=(0,1,\ldots,0), \ldots, P_{n+1}=(0,0,\ldots,1)$ of the canonical basis of $\mathbb{R} ^{n+1}$, the corresponding simplex
\begin{align*}
\Delta_n&:=\Delta_{P_{1},\ldots,P_{n+1}}\\
& = \left\{\left(x_{1}, x_{2}, \ldots, x_{n+1}\right)\in \mathbb{R}^{n+1}  \colon x_i \geq 0 \hbox{ for all $i=1,\ldots,n+1$ and } \sum_{i=1}^{n+1}x_{i} = 1 \right\}
\end{align*}
is called the \emph{standard $n$-simplex} and is denoted by $\Delta_{n}$.  Any $n$-simplex $\Delta_{P_{1},\ldots,P_{n+1}}$ can be canonically identified with the standard $n$-simplex through the bijective mapping
\begin{equation*}
\left(x_{1},\ldots, x_{n+1}\right) \in \Delta_{n} \mapsto \sum_{i=1}^{n+1}x_{i}P_{i} \in \Delta_{P_{1},\ldots,P_{n+1}}.
\end{equation*}
Thus, from now on we shall deal only with the standard $n$-simplex.  Notice that $\Delta_0$ is just the point $1 \in \mathbb{R}$, $\Delta_1$ the line segment in $\mathbb{R}^{2}$ joining $(1,0)$ to $(0,1)$, $\Delta_2$ the
equilateral triangle in $\mathbb{R}^3$ whose vertices are $(1,0,0)$, $(0,1,0)$, $(0,0,1)$, and $\Delta_3$ the regular tetrahedron in $\mathbb{R}^4$ with vertices $(1,0,0,0)$, $(0,1,0,0)$, $(0,0,1,0)$, $(0,0,0,1)$.

We point out that $\Delta_n$ is bijective to the set of ordered $(n+1)$-tuples
\begin{equation*}
\Delta_{n}^{\ast}:=\left\{\left(s_{1}, \ldots, s_{n}, s_{n+1} \right)\in \mathbb{R}^{n+1} : 0\leq s_{1} \leq \cdots \leq s_{n} \leq s_{n+1}=1 \right\}.
\end{equation*}
Indeed,  the map
\begin{gather} \label{bijection}
\varphi : \Delta_{n} \longrightarrow \Delta_{n}^{\ast}\\
\left(x_{1},\ldots,x_{n+1}\right) \mapsto \left(x_{1}, x_{1}+x_{2}, \ldots, x_{1}+x_{2}+\cdots+x_{n}, x_{1}+x_{2}+\cdots+x_{n+1}=1\right) \nonumber
\end{gather}
is clearly injective and surjective.  The facets of $\Delta_n$, which are given by the equation $x_{i}=0$, under this bijection correspond to successive coordinates being equal, $s_{i}=s_{i-1}$.

\section{A natural path toward the best assessment}\label{natural-path}
Let us consider a typical evaluation research call, where each institution is due to submit a certain number of publications depending on its size. Let us fix a hypothetical university department which has to submit $N$ products. According to the call's rules at the end of the evaluation  each product is assigned to a class of a predefined ordinal qualitative variable, with $n+1$  attributes, ranging between the possible best assessment (usually ``excellent'') and the worst one (usually ``poor'').

Let $x_i$ be the relative frequency of the number of publications assigned to the $i$th class. Since, for each $i\in\left\{1, \ldots, n+1\right\}$, $x_i \geq 0$ and $\sum_{i=1}^{n+1}x_i=1$, the global assessment of the department can be naturally identified with a point $(x_1, \ldots, x_{n+1})$ in the standard $n-$dimensional simplex $\Delta_n$. Notice that the best evaluation that the department can achieve is represented in the simplex by the point $P_1=(1,0, \ldots,0)$, corresponding to the ideal situation in which all the submitted publications are assessed in the best class. The remaining points in $\Delta_n$ represent intermediate assessments,  starting from $P_1$ until the worst evaluation represented by $P_{n+1}$. Thus, while, from a geometrical point of view, the simplex is a highly symmetric object, in our context the order of the vertices is very important.

Suppose that the final evaluation of the department is represented by the point $P_0=(x_1^0, \ldots, x_{n+1}^0)$. Then in this geometrical framework it is natural to try to measure how ``far'' is the point $P_0$ in the simplex from the vertex $P_1$ corresponding to the best possible assessment (see Figure~\ref{figura0}).
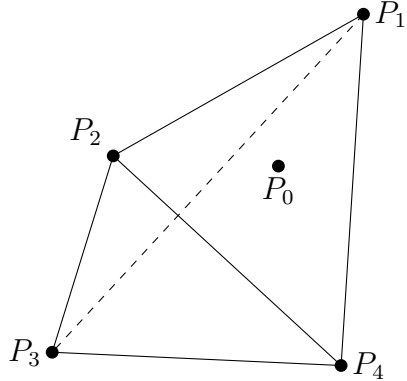
\begin{figure}[th!]
\begin{tikzpicture}
\begin{axis}[width=0.5\textwidth,
                axis lines=middle,
                inner axis line style={color=white},
                xmin=-1.4,
                xmax=1.7,
                ymin=-1.4,
                ymax=1.7,
                zmin=-1.4,
                zmax=1.7,
                xtick={0,6},
                ytick={0,6},
                ztick={0,6},
                view={105}{-5}]

\fill({1/5},{7/15},{1/3}) circle (2.4pt);

\node at ({1/5},{7/15},{1/3}) [below] {$P_0$};
\node at (-1.25299, -1.25299, 0.582337) [above left] {${P_2}$};
\node at (0.582337, -1.25299, -1.25299) [left] {${P_3}$};
\node at (-1.25299, 0.582337, -1.25299)[right] {${P_4}$};

\fill(1.5,1.5,1.5) circle (2.4pt);
\node at (1.5,1.5,1.5) [right] {$P_1$};
\draw(1.5,1.5,1.5) -- (-1.25299, -1.25299, 0.582337);
\draw[dashed](1.5,1.5,1.5)--(0.582337, -1.25299, -1.25299);
\draw(1.5,1.5,1.5)--(-1.25299, 0.582337, -1.25299);

\draw(-1.25299, -1.25299, 0.582337)--(-1.25299, 0.582337, -1.25299)--(0.582337, -1.25299, -1.25299);
\draw(-1.25299, -1.25299, 0.582337)--(0.582337, -1.25299, -1.25299);

\fill(-1.25299, -1.25299, 0.582337) circle (2.4pt);
\fill(-1.25299, 0.582337, -1.25299) circle (2.4pt);
\fill(0.582337, -1.25299, -1.25299) circle (2.4pt);

\end{axis}

\end{tikzpicture}
\caption{The final evaluation of the department is represented by the point $P_0=(x_1^0, \ldots, x_{n+1}^0)$ of the simplex. In the figure the case $n=3$.}\label{figura0}

\end{figure}

 Being the simplex a subset of the Euclidean space $\mathbb{R}^{n+1}$, in principle a natural choice would be to use the Euclidean distance. However in this way it could happen that departments with  very different evaluations  have the same distance from $P_1$ (for instance two distinct vertices $P_i$ and $P_j$ have the same Euclidean distance from $P_1$). Thus, since we are dealing with ordinal categorical variables, the way in which such measurement is done should take into account the ordering of the vertices of the simplex.  Any possible choice of such a ``distance'' between $P_0$ and $P_1$ should be defined in such a way to respect the order described above. Our idea is to construct a path from the point $P_0$ to $P_1$ so that the evaluation changes continuously through this path as slow as possible.

For exposition purposes we begin with the description of the $1$-dimensional case, in which only two categories are considered (the best and the worst assessments). Here the simplex degenerates into the line segment from the point $P_{1}=(1,0)$ to $P_{2}=(0,1)$. We express by a positive real number $a$ the ``effort'' for going from the worst evaluation $P_2$ to the best evaluation $P_1$. From a mathematical point of view we are declaring that the length of the vector $P_{2} - P_{1}$ is $a$. Then, given a point $P_0=(x_1^0, x_2^0)=(x^0_{1},1-x^0_{1})$ in such segment, a measurement $\delta(P_0)$ of the distance of $P_0$ from the best possible outcome $P_1$ is just given by the length of the segment from $P_0$ to $P_1$, that is (see Figure~\ref{figura1}~(a))

\begin{figure}[b]
\begin{tikzpicture}

\begin{axis}[width=0.45*\textwidth,
                axis lines=middle,
                xmin=-.5,
                xmax=1.5,
                ymin=-.5,
                ymax=1.5,
                xtick={0,2},
                ytick={0,2},
                ]

\draw[line width=1.2pt](1,0)--({1/3},{2/3});
\draw[decoration={markings,mark=at position 1 with
    {\arrow[scale=1.5,>=stealth]{>}}},postaction={decorate},color=red, line width=1.2pt,->,>=stealth] ({1/3},{2/3})--(0,1);

\fill({1/3},{2/3}) circle (2.4pt);
\fill(1,0) circle (2.4pt);
\fill(0,1) circle (2.4pt);

\node at (1,0) [below] {$P_2=(0,1)$};
\node at (0,1) [above right] {$P_1=(1,0)$};
\node at ({1/3},{2/3}) [above right] {$P_0=(x_1^0,x_2^0)$};

\end{axis}
\node at (3,-1) [above left] {(a)};
\end{tikzpicture}\hspace{10mm}
\begin{tikzpicture}
\begin{axis}[width=0.55*\textwidth,
                axis lines=middle,
                inner axis line style={dashed},
                xmin=0,
                xmax=1.5,
                ymin=0,
                ymax=1.5,
                zmin=0,
                zmax=1.5,
                xtick={0,2},
                ytick={0,2},
                ztick={0,2},
                view={135}{30}]

\addplot3[patch, color=blue, fill opacity=0.1, faceted color=black, line width=0.8pt] coordinates{(1,0,0) (0,1,0) (0,0,1)};

\draw[decoration={markings,mark=at position 1 with
    {\arrow[scale=1.5,>=stealth]{>}}},postaction={decorate},color=red, line width=1.2pt,->,>=stealth] ({1/5},{7/15},{1/3})--({2/3},0,{1/3});

\draw[decoration={markings,mark=at position 1 with
    {\arrow[scale=1.5,>=stealth]{>}}},postaction={decorate},color=red, line width=1.2pt,->,>=stealth]({2/3},0,{1/3})--(0,0,1);

\draw(1,0,0)--(1.5,0,0);
\draw(0,1,0)--(0,1.5,0);
\draw(0,0,1)--(0,0,1.5);

\fill(1,0,0) circle (2.4pt);
\fill(0,1,0) circle (2.4pt);
\fill(0,0,1) circle (2.4pt);

\fill({1/5},{7/15},{1/3}) circle (2.4pt);
\fill({2/3},0,{1/3}) circle (2.4pt);

\node at (1,0,0) [above left] {$P_2$};
\node at (0,1,0) [above right] {$P_3$};
\node at (0,0,1) [above right] {$P_1$};
\node at ({1/5},{7/15},{1/3}) [above right] {$P_0$};

\node at ({2/3},0,{1/3}) [above left] {$P_0^{'}$};

\end{axis}
\path (2.6,4.3) node(x) {$\ell_{12}$};
\path (3.1,3.2) node(x) {$\ell$};
\node at (4,-1) [above left] {(b)};
\end{tikzpicture}
\caption{Picture~(a) represents the 1-dimensional case: the length of the red segment is $\delta(P_0)$. Picture~(b) represents the 2-dimensional case: here  $\delta(P_0)$ is the sum of the length of the two red segments}\label{figura1}
\end{figure}
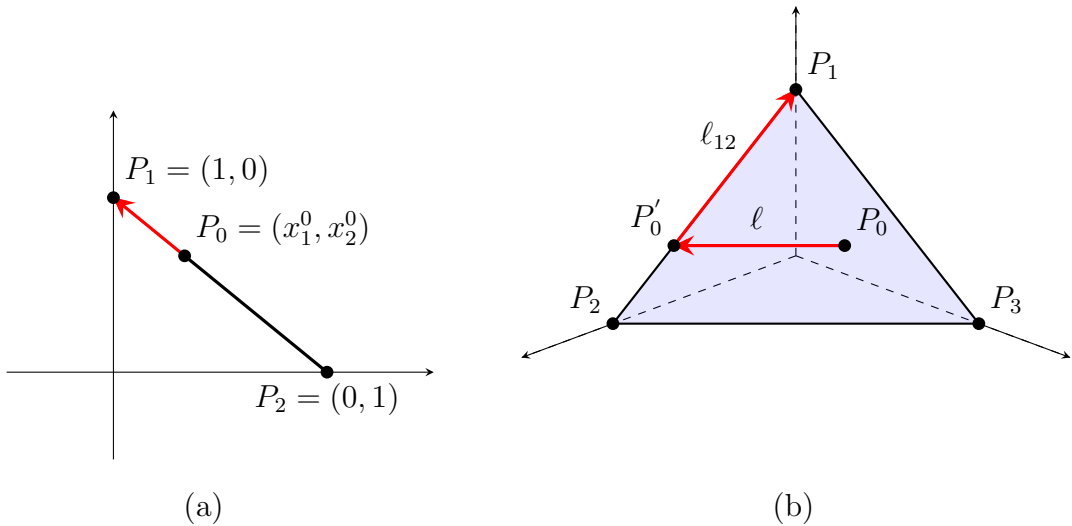

\begin{equation}\label{delta1}
\delta(P_0)=a(1-x_{1}^{0}).
\end{equation}
In the previous formula the choice $a=\sqrt{2}$ gives the standard Euclidean distance, while for $a=1$ \eqref{delta1} simplifies to $\delta(P_0)=1-x_{1}^{0}$.

Let us consider now the $2$-dimensional case, represented in Figure~\ref{figura1}~(b), corresponding to the case in which each publication can be evaluated in $3$ possible ways: the ``best'',  the ``intermediate'' and the ``worst'' ones. Unlike the $1$-dimensional case, here there is no natural order relation which can be used for giving an immediate measurement of how far $\delta(P_0)$ is from the best possible evaluation $P_1$. We proceed in the following way.  As before, let us  express  by two positive real numbers $a$ and $b$ the ``effort'' for going  from $P_{2}$ to $P_{1}$ and from $P_{3}$ to $P_{2}$, respectively. If $O$ denotes the origin in ${\mathbb R}^3$, the vectors $\mathbf{v}_{1}:=P_{1}- O = (1,0,0)$, $\mathbf{v}_{2}:=P_{2}- P_{1} = (-1,1,0)$ and $\mathbf{v}_{3}:=P_{3}- P_{2} = (0,-1,1)$ are linearly independent and thus they form a basis of $\mathbb{R}^3$. Let us consider the scalar product $g$ on $\mathbb{R}^3$, represented by the matrix
\begin{equation*}
M_{\mathcal{B}}(g)=
\begin{pmatrix}
1 & 0  & 0 \\
0 & a^2  & 0 \\
0 & 0 &   b^2
\end{pmatrix}
\end{equation*}
with respect to the basis $\mathcal{B}=\left\{ \mathbf{v}_{1}, \mathbf{v}_{2}, \mathbf{v}_{3} \right\}$. This scalar product induces a distance function $d_{g}$ on $\mathbb{R}^3$ in the usual way
\begin{equation*}
d_{g}(P,Q)=\left\| P-Q \right\|_{g} =  \sqrt{ g\left( P-Q , P-Q \right) }.
\end{equation*}
The restriction of $d_g$ to $\Delta_{2}$ provides the simplex with a distance, that will be still denoted by $d_g$. Let $P'_0$ be the intersection of the edge $\ell_{1 2}$ of the simplex through the points $P_1$ and $P_2$ with the line $\ell$  through $P_0=(x_1^0, x_2^0, x_3^0)$ with direction $\mathbf{v}_3$ (see Figure \ref{figura1}). If we restrict our attention to the line $\ell$, we recover the natural order relation in the  $1$-dimensional case previously considered. More precisely, as one moves from $P_0$ to $P'_0$ along $\ell$,  the corresponding outcome is ``improving'' in a natural way, as if we were continuously transferring part of publications that received the worst evaluation to the intermediate class. Once arrived at $P'_0$, we recover again the natural order provided by the geometry of the line $\ell_{1 2}$ connecting $P_1$ to $P_2$. Here, moving from $P'_0$  to $P_1$ corresponds to an increase of the frequency of publications in the best class and a consequent decrease of the frequency in  the intermediate class. In this way we are able to find a ``natural'' path connecting $P_0$ to $P_1$ (the red path in Figure  \ref{figura1}). Thus we can measure the ``distance'' $\delta(P_0)$ from $P_0$ to $P_1$ by the length of this path:
\begin{equation*}
\delta(P_0)=d_{g}(P_{0}, P'_{0}) + d_{g}(P'_{0}, P_{1}).
\end{equation*}
Since the line $\ell$ can be represented by the parametric equations
\begin{equation*}
\begin{cases}
x_1=x_1^0\\
x_2=x_2^0-t\\
x_3=x_3^0-t
\end{cases}
\end{equation*}
and the line $\ell_{12}$ has Cartesian equations
\begin{equation*}
\begin{cases}
x_1 + x_2 = 1\\
x_3 = 0
\end{cases}
\end{equation*}
the point $P'_0$ has coordinate $(x_{1}^0, 1  - x_{1}^0, 0)$ and so we find
\begin{align*}
\delta(P_0)&=\left\| P'_{0}-P_{0} \right\|_g + \left\| P_{1}-P'_{0} \right\|_g\\
&=\left\| -x_{3}^{0} \mathbf{v}_{3} \right\|_g + \left\| -(1-x_{1}^{0}) \mathbf{v}_{2}\right\|_g\\
&=x_{3}^{0} \sqrt{g(\mathbf{v}_{3},\mathbf{v}_{3})} + (1-x_{1}^{0}) \sqrt{g(\mathbf{v}_{2},\mathbf{v}_{2})}\\
&=(1-x_{1}^{0}-x_{2}^{0})b + (1-x_{1}^{0})a\\
&=(a+b) - (a+b)x_{1}^{0} - b x_{2}^{0}.
\end{align*}

If the evaluation call contemplates four different assessment classes, our model is encoded by the geometry of the $3$-dimensional simplex $\Delta_3$. Let $P_{0}=(x_{1}^{0},x_{2}^{0},x_{3}^{0},x_{4}^{0}) \in \Delta_3$ be the evaluation of the department. Then, generalising the above constructions, we are going to define a natural path joining $P_0$ to the best possible evaluation $P_1$. Here the ``naturality'' of such a path means that: (i) for any  $i,j \in\left\{2,3,4\right\}$ such that $i > j$ the path from $P_i$ to $P_1$ should be longer than the path from $P_j$ to $P_1$; (ii)  this property should be satisfied also by the vertices of any $2$-dimensional simplex given by the intersection of $\Delta_3$ with a plane whose direction is spanned by $P_{3}-P_{2}$ and $P_{4}-P_{2}$. Let $a$, $b$, $c$ denote positive real numbers expressing  the ``effort'' for going from $P_{2}$ to $P_{1}$,   from $P_{3}$ to $P_{2}$,  and from $P_{4}$ to $P_{3}$,  respectively. 
As in the previous case the vectors  $\mathbf{v}_{1}:=P_{1}- O$ and $\mathbf{v}_{i}:=P_{i+1}-P_{i}$, $i=1,2,3$, form a basis $\mathcal{B}$ of $\mathbb{R}^4$ and we consider the scalar product $g$ on $\mathbb{R}^4$, represented by the matrix
\begin{equation*}
M_{\mathcal{B}}(g)=
\begin{pmatrix}
1 & 0  & 0 & 0 \\
0 & a^2  & 0  & 0\\
0 & 0 &   b^2 & 0\\
0 & 0 &   0 & c^2
\end{pmatrix}\,.
\end{equation*}

This scalar product induces a distance function $d_{g}$ on $\mathbb{R}^4$ and  $\Delta_{3}$ becomes a metric space with distance the restriction of $d_{g}$ to ${\Delta_3}$. Let $\pi$ denote the plane through $P_0$ and parallel to  the plane containing the points $P_{2}, P_{3}, P_{4}$ (i.e. all the vertices of the simplex except the one corresponding to the best evaluation). Notice that $\pi$ has parametric equations
\begin{equation*}
\begin{cases}
x_1=x_1^0\\
x_2=x_2^0-t-s\\
x_3=x_3^0+t\\
x_4=x_4^0 + s
\end{cases}
\end{equation*}
The intersection of $\pi$ with the simplex $ \Delta_3$ is a $2$-dimensional simplex, the equilateral triangle with vertices $P'_{1}, P'_{2}, P'_{3}$. From the triangle obtained we can recover the $2$-dimensional construction.  The path from $P_0$ to $P_1$ is now determined by the union of 3 segments (see Figure~\ref{figura2}).
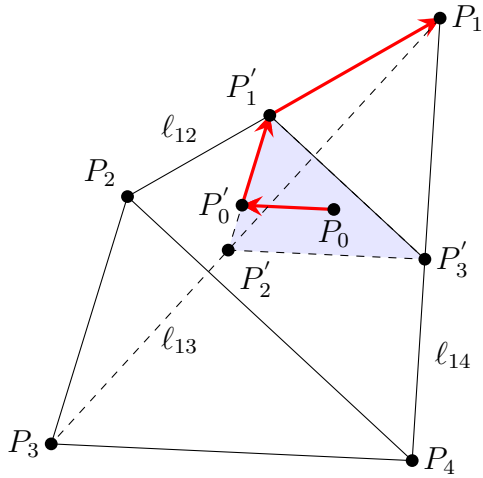
\begin{figure}
\begin{tikzpicture}
\begin{axis}[width=0.6\textwidth,
                axis lines=middle,
                inner axis line style={color=white},
                xmin=-1.4,
                xmax=1.7,
                ymin=-1.4,
                ymax=1.7,
                zmin=-1.4,
                zmax=1.7,
                xtick={0,6},
                ytick={0,6},
                ztick={0,6},
                view={105}{-5}
                ]

\addplot3[dashed, patch, color=blue, fill opacity=0.1, faceted color=black, line width=0.4pt] coordinates{(1,0,0) (0,1,0) (0,0,1)};

\draw[ decoration={markings,mark=at position 1 with
    {\arrow[scale=1.5,>=stealth]{>}}},postaction={decorate},color=red, line width=1.2pt,->,>=stealth] ({1/5},{7/15},{1/3})--({2/3},0,{1/3});

\draw(0,1,0)--(0,0,1);

\draw[ decoration={markings,mark=at position 1 with
    {\arrow[scale=1.5,>=stealth]{>}}},postaction={decorate},color=red, line width=1.2pt,->,>=stealth]({2/3},0,{1/3})--(0,0,1);

\draw[decoration={markings,mark=at position 1 with
    {\arrow[scale=1.5,>=stealth]{>}}},postaction={decorate},color=red, line width=1.2pt,->,>=stealth](0,0,1)--(1.5,1.5,1.5);


\fill(1,0,0) circle (2.4pt);
\fill(0,1,0) circle (2.4pt);
\fill(0,0,1) circle (2.4pt);

\fill({1/5},{7/15},{1/3}) circle (2.4pt);
\fill({2/3},0,{1/3}) circle (2.4pt);

\node at (1,0,0) [below right] {$P_2^{'}$};
\node at (0,1,0) [right] {$P_3^{'}$};
\node at (0,0,1) [above left] {$P_1^{'}$};
\node at ({1/5},{7/15},{1/3}) [below] {$P_0$};

\node at ({2/3},0,{1/3}) [left] {$P_0^{'}$};

\node at (-1.25299, -1.25299, 0.582337) [above left] {${P_2}$};
\node at (0.582337, -1.25299, -1.25299) [left] {${P_3}$};
\node at (-1.25299, 0.582337, -1.25299)[right] {${P_4}$};

\fill(1.5,1.5,1.5) circle (2.4pt);
\node at (1.5,1.5,1.5) [right] {$P_1$};
\draw(0,0,1) -- (-1.25299, -1.25299, 0.582337);
\draw[dashed](1.5,1.5,1.5)--(0.582337, -1.25299, -1.25299);
\draw(1.5,1.5,1.5)--(-1.25299, 0.582337, -1.25299);

\draw(-1.25299, -1.25299, 0.582337)--(-1.25299, 0.582337, -1.25299)--(0.582337, -1.25299, -1.25299);
\draw(-1.25299, -1.25299, 0.582337)--(0.582337, -1.25299, -1.25299);

\fill(-1.25299, -1.25299, 0.582337) circle (2.4pt);
\fill(-1.25299, 0.582337, -1.25299) circle (2.4pt);
\fill(0.582337, -1.25299, -1.25299) circle (2.4pt);

%
%


\end{axis}
\path (2.6,4.8) node(x) {$\ell_{12}$};
\path (6.2,1.8) node(x) {$\ell_{14}$};
\path (2.6,2.) node(x) {$\ell_{13}$};
\end{tikzpicture}
\caption{Representation of the 3-dimensional case: here $\delta(P_0)$ is the sum of the length of the three red segments}\label{figura2}

\end{figure}
Since  the line $\ell_{i j}$ joining $P_i$ with $P_j$, $i,j\in\left\{1,2,3\right\}$, has equations
\begin{equation*}
\begin{cases}
x_i + x_j = 1\\
x_k = 0 	 \ \hbox{ for any } k\notin \left\{i,j\right\}
\end{cases}
\end{equation*}
we find that the coordinates of the vertices of the $2$-dimensional simplex are $P'_{1}=(x_{1}^{0}, 1-x_{1}^{0},0,0)$, $P'_{2}=(x_{1}^{0}, 0,1-x_{1}^{0},0)$, $P'_{3}=(x_{1}^{0},0, 0,1-x_{1}^{0})$. In order to find the coordinates of $P'_0$, let $\ell'_{12}$ and $\ell$ denote, respectively, the line joining $P'_{1}$ with $P'_{2}$ and the line through $P_0$ with direction  $P'_{3}-P'_{2} $. Such two coplanar lines have equations
\begin{equation*}
\ell'_{12}:
\begin{cases}
x_1 = x_{1}^{0}\\
x_2 + x_3 = 1 - x_{1}^{0} \\
x_4 = 0 	
\end{cases} \quad \quad  \ell: \begin{cases}
x_1=x_1^0\\
x_2=x_2^0\\
x_3=x_3^0-(1-x_{1}^{0}) t\\
x_4=x_4^0 + (1 - x_{1}^{0})t
\end{cases}
\end{equation*}
Thus $P'_{0}= \ell'_{12} \cap \ell= (x_{1}^{0}, x_{2}^{0}, 1 - x_{1}^{0}-x_{2}^{0}, 0)$. Hence

\begin{align*}
\delta(P_0)&=d_{g}(P_{0}, P'_{0}) + d_{g}(P'_{0}, P'_{1}) + d_{g}(P'_{1}, P_{1})\\
&=\left\| P'_{0}-P_{0}  \right\|_g + \left\| P'_{1}-P'_{0}  \right\|_g + \left\| P_{1}-P'_{1}  \right\|_g \\
&=\left\|-x_{4}^{0} \mathbf{v}_{4}  \right\|_g + \left\| -(1-x_{1}^{0}-x_{2}^{0})\mathbf{v}_{3}  \right\|_g + \left\| -(1-x_{1}^{0})\mathbf{v}_2  \right\|_g \\
&=x_{4}^{0} \sqrt{g(\mathbf{v}_{4},\mathbf{v}_{4})} + (1-x_{1}^{0}-x_{2}^{0}) \sqrt{g(\mathbf{v}_{3},\mathbf{v}_{3})} + (1-x_{1}^{0})\sqrt{g(\mathbf{v}_{2},\mathbf{v}_{2})}\\
&=(1-x_{1}^{0}-x_{2}^{0}-x_{3}^{0}) c +  (1-x_{1}^{0}-x_{2}^{0})b + (1-x_{1}^{0})a\\
&=(a+b+c) - (a+b+c)x^{0}_{1} - (b+c)x_{2}^{0} - c x_{3}^{0}\,.
\end{align*}

In the general case, when one has $n+1$ categories, iterating the above constructions, the defined path joining $P_0=(x_{1}^{0},\ldots,x_{n+1}^{0})$ with $P_1$ can be obtained as the union of $n$ line segments, each of which lying in an
$(n-1)$-dimensional simplex. Thus we we finally obtain the following formula
\begin{equation}\label{eq:delta-n}
\delta(P_0)= a_{1}+\cdots + a_{n} - (a_{1}+\cdots + a_{n})x_{1}^{0} - (a_{2}+\cdots + a_{n}) x_{2}^{0} - \cdots - (a_{n-1}+a_{n}) x_{n-1}^{0} - a_{n} x_{n}^{0}
\end{equation}
where, for each $i \in \left\{1,\ldots, n\right\}$, $a_i$ is a positive real number expressing the ``effort'' for going from the vertex $P_{i+1}$ to $P_{i}$.

A rigorous proof of \eqref{eq:delta-n} can be done by induction and it is reported in the Appendix.

\begin{remark}
The constants $a_i$ in \eqref{eq:delta-n}, i.e. the quantification of the effort for going from a category to that one immediately higher, should be explicated in the call. The most frequent situation is when all such efforts are considered equivalent, so that the constants can be taken all equal to $1$. In this case \eqref{eq:delta-n} simplifies  to
\begin{equation}\label{eq:delta-n-bis}
\delta(P_0)= n - n x_{1}^{0} - (n-1) x_{2}^{0} - \cdots - 2 x_{n-1}^{0} -  x_{n}^{0}.
\end{equation}
In this particular case, $\delta(P_0)$ can be expressed also in terms of the Minkowski distance $d_M$, often applied in measuring dissimilarity of ordinal data - some recent application in this direction can be found in \cite{Weiss2019} and \cite{Weiss2020}. In this regard, using the bijection \eqref{bijection}, we have
\begin{align*}
\delta(P_0)&= n - n x_{1}^{0} - (n-1) x_{2}^{0} - \cdots - 2 x_{n-1}^{0} -  x_{n}^{0}\\
&=(1 - x_{1}^{0}) + (1 - x^{0}_{1} - x^{0}_{2}) + \cdots +   (1 - x^{0}_{1} - x^{0}_{2} - \cdots  - x^{0}_{n})\\
&=|x_{1}^{0} - 1| + | x^{0}_{1} + x^{0}_{2} - 1| + \cdots +   | x^{0}_{1} + x^{0}_{2} + \cdots  + x^{0}_{n} - 1|\\
&=d_{M}( (x^{0}_{1}, x^{0}_{1} + x^{0}_{2}, \ldots, x^{0}_{1} + x^{0}_{2} + \cdots + x^{0}_{n} , 1), (1, 1, \ldots, 1, 1))\\
&=d_{M}(\varphi(x^{0}_{1}, x^{0}_{2}, \ldots, x^{0}_{n}, x^{0}_{n+1}), \varphi(1, 0, \ldots, 0, 0))\\
&=d_{M}(\varphi({P}_{0}), \varphi({P}_{1}))\,.
\end{align*}

However, there could be - and actually there were - situations when the assumption $a_1=\cdots a_n = 1$ can not be necessarily taken. One example is the VQR 2011-14 which will be discussed in Section \ref{VQR}.
\end{remark}


\section{Geometric score function}
Starting from $\delta(P_0)$, see equation \eqref{eq:delta-n}, we can  naturally define a map
\begin{equation*}
d : \Delta_{n} \times \Delta_{n} \longrightarrow \mathbb{R}
\end{equation*}
such that for any $P_{0}, Q_{0} \in \Delta_{n}$
\begin{equation*}
d(P_{0},Q_{0}):=\mid \delta(P_0) - \delta(Q_0) \mid.
\end{equation*}
In other words, we are comparing the evaluations $P_0$ and $Q_0$ of two departments, measuring how ``far'' is each one from $P_1$. Note that $d$ is clearly non-negative and symmetric. Moreover, it satisfies the triangular inequality, since
\begin{align*}
d(P,P'')&=\mid \delta(P) - \delta(P'') \mid \\
&= \mid \delta(P) - \delta(P') + \delta(P')- \delta(P'') \mid \\
& \leq \mid \delta(P) - \delta(P') \mid  + \mid \delta(P')- \delta(P'') \mid \\
&=d(P,P') + d(P',P'').
\end{align*}
Note that $d$ is a pseudo-distance, since it does not satisfies the \emph{identity of indiscernibles} condition. In fact $d(P,P')=0$ does not necessarily imply  that $P=P'$. For instance, in the $2$-dimensional case, the points $P=(\frac{1}{4},\frac{3}{4},0)$ and $P'=(\frac{1}{2},\frac{1}{4},\frac{1}{4})$ are  such that $\delta(P)=\delta(P')=\frac{3 \sqrt{2}}{4}$, so that $d(P,P')=0$.

\medskip

In the applications, in order to compare the assessments of different departments or other aggregates, it is important to consider   the locus of points  of the simplex $\Delta_{n}$ that are at distance $0$ from each other.  Geometrically, in view of \eqref{eq:delta-n}, such a set can be described  by the sheaf of parallel hyperplanes of equation
\begin{equation}\label{hyperplane}
(a_{1}+\cdots + a_{n})x_{1}^{0} +(a_{2}+\cdots + a_{n}) x_{2}^{0} + \cdots + (a_{n-1}+a_{n}) x_{n-1}^{0} + a_{n} x_{n}^{0} = \textrm{const}.
\end{equation}
This allows us to divide the aggregates under study in equivalent classes, corresponding to such loci. More formally, one can consider the relation $\sim$ on $\Delta_{n}$ that identifies any two points $P$ and $P'$ such that $d(P,P')=0$. It can be easily proved that $\sim$ is an equivalence relation and then $d$ turns out to be a distance on the quotient set $\Delta_{n} / \sim$. Thus $\Delta_n$ can be partitioned into equivalence classes, which correspond to the points of the simplex belonging to each hypeplane \eqref{hyperplane} (see Figure~\ref{figura4}~(a)).

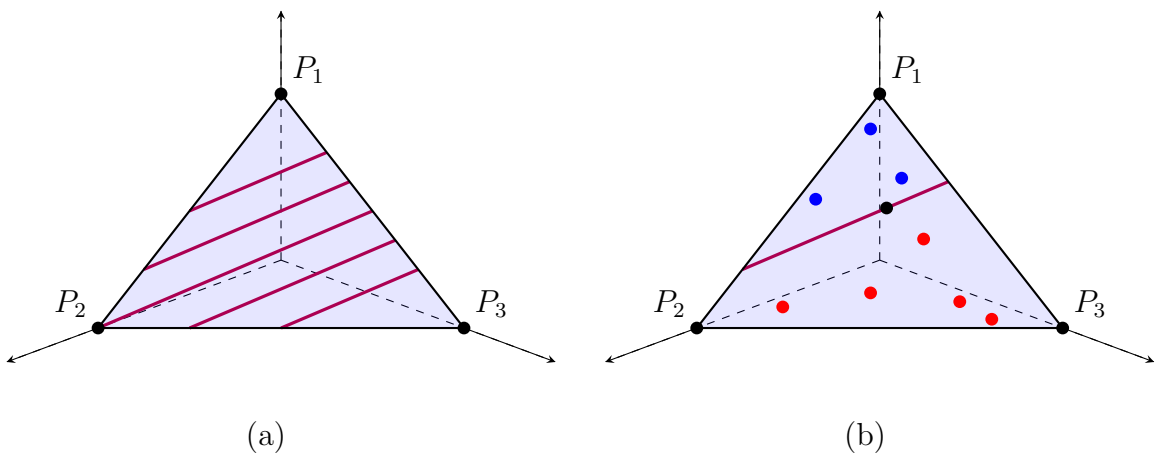
\begin{figure}[b]
\begin{tikzpicture}
\begin{axis}[width=0.55*\textwidth,
                axis lines=middle,
                inner axis line style={dashed},
                xmin=0,
                xmax=1.5,
                ymin=0,
                ymax=1.5,
                zmin=0,
                zmax=1.5,
                xtick={0,2},
                ytick={0,2},
                ztick={0,2},
                view={135}{30}]

\draw[color=purple, line width=1.2pt](1,0,0)--(0,.5,.5);
\draw[color=purple, line width=1.2pt](0.5, 0., 0.5)--(0., 0.25, 0.75);

\draw[color=purple, line width=1.2pt](0.75, 0., 0.25)--(0., 0.375, 0.625);
\draw[color=purple, line width=1.2pt](0.75, 0.25, 0.)--(0., 0.625, 0.375);
\draw[color=purple, line width=1.2pt](0.5, 0.5, 0.)--(0., 0.75, 0.25);
\addplot3[patch, color=blue, fill opacity=0.1, faceted color=black, line width=0.8pt] coordinates{(1,0,0) (0,1,0) (0,0,1)};

\draw(1,0,0)--(1.5,0,0);
\draw(0,1,0)--(0,1.5,0);
\draw(0,0,1)--(0,0,1.5);

\fill(1,0,0) circle (2.4pt);
\fill(0,1,0) circle (2.4pt);
\fill(0,0,1) circle (2.4pt);

\node at (1,0,0) [above left] {$P_2$};
\node at (0,1,0) [above right] {$P_3$};
\node at (0,0,1) [above right] {$P_1$};

\end{axis}
\node at (3,0) [above right] {(a)};

\end{tikzpicture}\hspace{5mm}
\begin{tikzpicture}
\begin{axis}[width=0.55*\textwidth,
                axis lines=middle,
                inner axis line style={dashed},
                xmin=0,
                xmax=1.5,
                ymin=0,
                ymax=1.5,
                zmin=0,
                zmax=1.5,
                xtick={0,2},
                ytick={0,2},
                ztick={0,2},
                view={135}{30}]

\draw[color=purple, line width=1.2pt](0.75, 0., 0.25)--(0., 0.375, 0.625);
\addplot3[patch, color=blue, fill opacity=0.1, faceted color=black, line width=0.8pt] coordinates{(1,0,0) (0,1,0) (0,0,1)};

\draw(1,0,0)--(1.5,0,0);
\draw(0,1,0)--(0,1.5,0);
\draw(0,0,1)--(0,0,1.5);

\fill(1,0,0) circle (2.4pt);
\fill(0,1,0) circle (2.4pt);
\fill(0,0,1) circle (2.4pt);

\fill(0.225, 0.2625, 0.5125) circle (2.4pt);

\fill[color=red](0.175, 0.7875, 0.0375) circle (2.4pt);
\fill[color=red](0.225, 0.6625, 0.1125) circle (2.4pt);
\fill[color=red](0.45, 0.4, 0.15) circle (2.4pt);
\fill[color=red](0.19, 0.43, 0.38) circle (2.4pt);
\fill[color=red](0.72, 0.19, 0.09) circle (2.4pt);

\fill[color=blue](0.4, 0.05, 0.55) circle (2.4pt);
\fill[color=blue](0.12, 0.24, 0.64) circle (2.4pt);
\fill[color=blue](0.1, 0.05, 0.85) circle (2.4pt);

\node at (1,0,0) [above left] {$P_2$};
\node at (0,1,0) [above right] {$P_3$};
\node at (0,0,1) [above right] {$P_1$};

\end{axis}
\node at (3,0) [above right] {(b)};
\end{tikzpicture}
\caption{In Picture~(a) a line of the pencil represents points with the same score function. In Picture~(b) a geometric interpretation of ${\mathcal S}_{\mathcal{C}}(A)$}\label{figura4}
\end{figure}

\medskip

We introduce a \emph{score function}  for the evaluation  of an aggregate $A$ (a department, a scientific area, etc.) in the following way: the assessment of $A$ can be realized as a point $P_A$ of $\Delta_n$, to which we can associate the number $\delta(P_A)$. This is, of course, an ``absolute'' index, in the sense that it can be used for comparing  evaluations among homogenous aggregates. For instance, it can be used to assign research fundings or to compare the quality of research of different candidates in a competition within the same scientific discipline. More in general, its usage can go beyond the context of this paper, i.e. research evaluation, since in principle it can be used whenever one deals with a situation where there are ordinal assessments. 

However, the above index can not in principle be appropriate to compare two inhomogeneous situations, and in particular the evaluations of departments, which are usually composed of   researchers working in different scientific areas.  In order to overcome this problem, we propose the following general approach. Let us fix an aggregate $A$ within a certain research evaluation call. We denote by $\mathfrak{I}(A,\mathcal{C})$  the set of all  ``ideal'' aggregates whose size and configuration (with respect to some prefixed criteria $\mathcal{C}$) are the same of $A$, and whose members are randomly selected from the set of all researchers that satisfy $\mathcal{C}$, working in any other university partecipating to the same assessment call. To any element of $\mathfrak{I}(A,\mathcal{C})$ it can be associated a point in the simplex. Then we define  the \emph{geometric score} of $A$ as
\begin{equation}\label{geometric-score}
{\mathcal S}_{\mathcal{C}}(A)=\frac{|\mathfrak{I}^{-}(A,\mathcal{C})|}{|\mathfrak{I}(A,\mathcal{C})|}
\end{equation}
where $|\mathfrak{I}(A,\mathcal{C})|$ denotes the cardinality of the set $\mathfrak{I}(A,\mathcal{C})$ and $|\mathfrak{I}^{-}(A,\mathcal{C})|$ denotes the number of ideal aggregates in $\mathfrak{I}(A,\mathcal{C})$ that are represented geometrically as  the  points of the simplex which are lying below the hyperplane \eqref{hyperplane} determined by $A$ (see Figure~\ref{figura4}~(b)).  ${\mathcal S}_{\mathcal{C}}(A)$ represents the probability that an  ``ideal'' aggregate $A'$, with the same configuration as $A$ (hence comparable with $A$), performs worse than $A$.

The choice of the conditions $\mathcal{C}$ depends on the availability of the data and should be aimed to make the evaluation as homogeneous as possibile. For instance, $\mathcal{C}$ may consist in the requirements that an ideal element of $\mathfrak{I}(A,\mathcal{C})$ must be composed of the same number as $A$ of researchers belonging to a given scientific area, and/or by the same number as $A$ of full / associate / assistant professors and/or the same proportion as $A$ of male and female researchers. Other examples for $\mathcal{C}$ may involve further information on researchers working in $A$, like teaching duties, salary, age, as well as description of the social and economic context where $A$ operates. In principle, the more  the possible choices for $\mathcal{C}$ are various and
precise, the more the evaluation is accurate and non-misleading. However, even in the extreme case when $\mathcal{C}=\emptyset$ the geometric score \eqref{geometric-score} is still informative, since it  avoids comparisons between aggregates of different size.

From the definition \eqref{geometric-score} and   the above considerations, it is clear that the geometric score of $A$  strongly depends on $\mathcal{C}$. This reflects mathematically the fact that there does not exist  ``the''  evaluation of $A$, but in fact there are many possible evaluations, each of them depending on the contextual aspects and on the possible refinements that  the user (the policy maker, the academician, the future student, the citizen, etc.) is interested to.

Operatively, the institution responsible of the evaluation (in most cases the national evaluation agency) should make available a wide range of information regarding the researchers involved in the evaluation as well as  the universities where they work. The utopia could consist in having access to a web-site where the user, according to his/her objectives, may select the most appropriate information that form the criteria $\mathcal{C}$, and consequently compute the corresponding geometric score.

\section{A case study: Italian research assessment VQR 2011-14}\label{VQR}

As an application, we consider the outcomes of the Italian Research Assessment VQR 2011-2014 within two scientific areas, the area of \emph{Mathematics} and the area of
 \emph{Statistics and Mathematical Methods for Decisions}  (coded, respectively, 01/A and 13/D according to the Italian scientific disciplines codification).  These two areas are made of a number of smaller aggregates known in the Italian system as Disciplinary Scientific Sectors (SSDs), as described in Table \ref{descrizione-01A} and Table \ref{descrizione-13D}.



\begin{table}[htbp]
  \centering
\tiny
  \caption{Composition of the area 01/A - Mathematics according to Italian Higher Education legislation}\label{descrizione-01A}
    \begin{tabular}{|l|l|}
    \hline
    {SSD Code} & {Description} \\
    \hline
    MAT/01 & Mathematical logic       \\
MAT/02 & Algebra       \\
MAT/03 & Geometry        \\
MAT/04 & Mathematics education and history of mathematics        \\
MAT/05 & Mathematical analysis          \\
MAT/06 & Probability and mathematical statistics         \\
MAT/07 & Mathematical physics           \\
MAT/08 & Numerical analysis \\
MAT/09 & Operational research           \\
    \hline
    \end{tabular}
\end{table}

\medskip

\begin{table}[htbp]
  \centering
\tiny
  \caption{Composition of the area 13/D - Statistics and Mathematical Methods for Decisions according to Italian Higher Education legislation}\label{descrizione-13D}
    \begin{tabular}{|l|l|}
    \hline
    {SSD Code} & {Description} \\
    \hline
    SECS-S/01 & Statistics       \\
SECS-S/02 & Statistics for experimental and technological research      \\
SECS-S/03 & Economic statistics        \\
SECS-S/04 & Demography        \\
SECS-S/05 & Social statistics  \\
SECS-S/06 & Mathematical methods for economics, actuarial and financial sciences       \\
    \hline
    \end{tabular}
\end{table}

The choice of these two areas is due to the different criteria adopted by the respective committees for the evaluation of the publications.

We recall that in the 2011-2014 VQR edition each researcher was expected to submit a given number of products (in most cases $2$). The submitted scientific products were classified by a committee  in one among the following classes, each one associated to a score:  Excellent - A (score: 1), Good - B (score: 0.7), Fair - C (score: 0.4), Acceptable - D (score: 0.1), Limited or not assessable - E (score: 0).

Then to each university   a so-called ``normalized average score'', denoted by $R$, was associated and was used by ANVUR to draw up a ranking. We recall that the $R$ score was computed as the ratio between the average score of the researchers of a given university in a given area / SSD and the average score of all the Italian researchers in that area / SSD (see \cite{ANVUR:2017-2}). A $R$ score greater than $1$ indicates that, in a given area / SSD,  the university under consideration performs better than the average of the all Italian universities. The final ranking, for all scientific areas, is available on the ANVUR's web-site (\cite{ANVUR:2017-2}).

Notice that the value of $R$ for each university, and hence the corresponding ANVUR ranking, is strongly linked to the score associated to each class, which even when reasonable, it is still arbitrary. Moreover, $R$ depends on the internal composition of the aggregate inside each university, which can considerably differ among universities. For instance, as for the 13/D area, the percentage repartition in the $6$ SSDs described in Table \ref{descrizione-13D} is $40$, $4$, $13$, $7$, $4$, $33$ for Roma ``La Sapienza'', while it is $53$, $0$, $0$, $0$, $0$, $47$ for the  Milano Politecnico. Last, $R$ is sensible to the size of the aggregate. This can generate some bias due to the so-called ``funnel effect'' (\cite{Spiegelhalter}).
\medskip

We have done a parallel ranking using our geometric score.

In order to apply the methods discussed in the previous sections, we need to understand how to choose, in this circumstance, the constants $a_{i}$ of formula \eqref{eq:delta-n}. Such positive numbers should correspond to the effort for moving from one category to the upper one, and should be known by the evaluating committee / referees. Unfortunately in this case there is not a univocal answer. A first option is to take into account the description of each class in the VQR call\footnote{\url{https://www.anvur.it/attivita/vqr/vqr-2011-2014/riferimenti-normativi-e-regolamentari/}}, which we have summarised in Table~\ref{VQR-call}. Indeed the call states that a publication should be considered \emph{excellent} if, ideally, it falls in the highest 10\% of the distribution of the international scientific research production of the Area in the period 2011-2014, \emph{good} if it is in the 10-30\% segment, \emph{fair} if it is in the 30-50\% segment, \emph{acceptable} if it is in the 50-80\% segment, \emph{limited or not assessable} if it is in the 80-100\% segment.
 Then a possible choice for $a_i$ could be to consider the effort for an upgrade to the higher class, measured as the distance between the lowest limit of each segment and the lowest limit of the next upper category.
Namely,  if we put $a_{4}=1$,  since the length of the interval corresponding to the class D is $1.5$-times longer than the ones of the classes B and C, one should have  $a_{1}=a_{2}=a_{4}=1$ and $a_{3}=1.5$.

\begin{table}[htbp]
  \centering
\tiny
  \caption{Description of each category in the VQR 2011-14 call and scores assigned by ANVUR to each class}\label{VQR-call}
    \begin{tabular}{|l|cc|}
    \hline
    Category & {Percentile in the distribution of worldwide publications in the area }  & {ANVUR scores} \\
    \hline
    A (Excellent) & 90--100   & 1         \\
B (Good) & 70--90   & 0.7         \\
C (Fair) & 50--70   & 0.4         \\
D (Aceptable) & 20--50   & 0.1         \\
E (Limited or not assessable) & 0--20   & 0        \\
    \hline
    \end{tabular}
\end{table}


On the other hand, one can argue that a referee is aware of the scores assigned by ANVUR to each publication falling in a certain class and used for computing the $R$ score. Thus another way for determining $a_i$ is to consider the difference between any two consecutive of such scores, obtaining $a_{1}=a_{2}=a_{3}=3$, $a_{4}=1$.

Finally, according to the codification and qualitative description used by ANVUR for the  assessment of the scientific products,   the referee may perceive the  classes as ``equidistant''. In this case we can assume $a_{i}=1$ for each $i\in\left\{1,\ldots,4\right\}$.

It is then interesting to test our geometric score in all these three situations and to compare the results with the official ANVUR outcomes. We fix as $\mathcal{C}$ the property  ``each ideal aggregate should have the same composition, with respect to each SSD, as the university under consideration'' .
Since for privacy reasons we can not associate the assessment of a research product to its author, we can not proceed by sampling directly from the set of all researchers in the area under consideration. Then, starting from the internal composition of the areas 01/A and 13/D, for a given university $U$ we have simulated an ideal aggregate with the same number of expected products in each SSD as $U$ and we have computed the corresponding score function. In principle one should compute the score function $\delta$ for all possible ideal aggregates that one can construct. However, the number of such ideal aggregates can be quite huge. For instance, for the 01/A Area of Roma ``La Sapienza''  (the biggest Italian university), this number is $43.65\times 10^{401}$.

To avoid to consider the enormous number of all possible combinations, we calculate the geometric score for each university in the two scientific areas by means of Monte Carlo simulations.


The algorithm used follows a simple scheme.  Once fixed the university $U$,  we denote with $m_i$ the number of product of the SSDs MAT/0i,  $i=1, \ldots,9$ (or alternatively, for the 13/D Area, SECS-S/0i, $i=1, \ldots,6$).  

We start with simulating $m_i$ products for the SSD MAT/0i belonging to one of the $5$ categories (from A to E) as a sample from a multivariate hypergeometric distribution with parameters $m_{i j}, j=1, \ldots,5$, the number of publications of Italian researchers in the SSD MAT/0i belonging to the category $j$ (see Tables \ref{VQR-math}--\ref{VQR-stat}). We then obtain the configuration of an ``ideal aggregate'' for which we can compute the value of $\delta$. We repeat such procedure $N$ times  and we obtain $N$ values of the  indexes $\delta_l, l=1, \ldots,N$.  We then compare all values obtained with that of the university $U$. This allows us to calculate the proportions of ideal aggregates performing worse than $U$. Since for each $l$ the random variable  $I_{\{\delta_l > \delta(U)\}}$ ($I_A$ being the indicator function of the set $A$)  follows a Bernoulli distribution with probability of success equal to  ${\mathcal S}_{\mathcal{C}}(U)$,   by the strong law of large numbers we can then approximate the value of our geometric score with   $ \frac{\sum_{l=1}^{N} I_{\{\delta_l > \delta(U)\}}}{N}$.  Then we can determine $N$ imposing that $P\Big( \frac{\sum_{l=1}^{N} I_{\{\delta_l > \delta(U)\}}}{N}-{\mathcal S}_{\mathcal{C}}(U) \geq 0.005 \Big)$  is very small (for instance of the order of $10^{-5}$). Using the Hoeffding's inequality (\cite{Hoeffding: 1963}) we have
$$P\Big( \frac{\sum_{l=1}^{N} I_{\{\delta_l > \delta(U)\}}}{N}-{\mathcal S}_{\mathcal{C}}(U) \geq 0.005 \Big) \leq e^{-2N\cdot 0.005^2 }.$$  For $N=200000$, that we fix as the number of simulations in all cases considered,  we have    $e^{-2N\cdot 0.005^2 }=4.539993 \cdot 10^{-5}$.   

We repeat such procedure for all universities in the area 01/A and then  in  the area 13/D. 

The  rankings obtained for the areas of Mathematics (01/A) and Statistics and Mathematical Methods for Decisions (13/D) are reported in Table \ref{tab:01Arank} and in Table \ref{tab:13Drank}, respectively. We have reported the official results of VQR 2011-2014 evaluation and compared it with the results arising by the application of the new method proposed considering the calculation of $\delta$ with the aforementioned three ways of choices of the constant $a_i$.

From the tables we can see how the rankings obtained with our score are quite different to that obtained from the $R$ score used by ANVUR for the VQR assessment. In order to verify the magnitude of the differences obtained, we applied a Kruskal-Wallis non-parametric test  to globally compare the four rankings. It emerged that the rankings were statistically different (p-value $ =4.627\cdot 10^{-13}$ for the area of Mathematics, and p-value $=1.377 \cdot 10^{-9}$ for the area of Statistics and Mathematical Methods for Decision). A pairwise comparison, through a post-hoc statistical test based on Wilcoxon statistic, confirmed that there is a difference between the ranking based on $R$ score and each of the rankings determined through ${\mathcal S}_{\mathcal{C}}$ (p-values were of order $10^{-10}$ when considering the scientific area 01/A and of order $10^{-7}$ for the area 13/D). On the other hand there are no relevant differences when comparing between them the rankings obtained from the three versions of ${\mathcal S}_{\mathcal{C}}$ (p-values above $0.94$ in both scientific areas).

\begin{table}[htbp]
  \centering
\tiny
    \begin{tabular}{|l|cc|ccc|ccc|ccc|}\hline
 & & & \multicolumn{3}{|c|}{(A)} & \multicolumn{3}{|c|}{(B)} & \multicolumn{3}{|c|}{(C)}\\\hline
 & &       &       &       &       &       &       &       &       &       &  \\
University & $R$ & Rank & ${\mathcal S}_{\mathcal{C}}(U)$ & Rank    & $\delta(U)$ &${\mathcal S}_{\mathcal{C}}(U)$ &  Rank         & $\delta(U)$       & ${\mathcal S}_{\mathcal{C}}(U)$& Rank & $\delta(U)$ \\
          &       &       &       &       &       &       &       &       &       &       &  \\
    \hline
  
    Pisa Normale & 1.4825 & 1     & 0.99995 & 3     & 0.37  & 0.99996 & 3     & 1.11  & 0.99997 & 3     & 0.37 \\
    Trieste SISSA & 1.4747 & 2     & 1     & 1     & 0.42  & 1.00000 & 1     & 1.14  & 1.00000 & 1     & 0.47 \\
    Pavia & 1.3733 & 3     & 1     & 1     & 0.59  & 1.00000 & 1     & 1.71  & 1.00000 & 1     & 0.61 \\
    Bergamo & 1.3077 & 4     & 0.919075 & 10    & 0.75  & 0.91712 & 12    & 2.11  & 0.91378 & 10    & 0.79 \\
    Brescia & 1.2778 & 5     & 0.99351 & 6     & 0.83  & 0.99514 & 6     & 2.33  & 0.99372 & 6     & 0.88 \\
    Cassino & 1.2576 & 6     & 0.88475 & 13    & 0.81  & 0.86903 & 14    & 2.43  & 0.84009 & 15    & 0.95 \\
    Verona & 1.2469 & 7     & 0.965495 & 9     & 0.91  & 0.97032 & 9     & 2.57  & 0.96872 & 9     & 0.95 \\
    Roma Tre & 1.2342 & 8     & 0.999885 & 4     & 0.88  & 0.99977 & 5     & 2.56  & 0.99983 & 4     & 0.94 \\
    Roma Tor Vergata & 1.1774 & 9     & 0.999605 & 5     & 1.05  & 0.99979 & 4     & 2.93  & 0.99939 & 5     & 1.13 \\
    Torino Politecnico & 1.1554 & 10    & 0.976015 & 8     & 1.1   & 0.97803 & 8     & 3.1   & 0.97619 & 7     & 1.19 \\
    Bari Politecnico & 1.1520 & 11    & 0.837835 & 15    & 1.1   & 0.83288 & 15    & 3.12  & 0.85786 & 14    & 1.16 \\
    Pisa  & 1.1204 & 12    & 0.977295 & 7     & 1.17  & 0.98849 & 7     & 3.23  & 0.97460 & 8     & 1.27 \\
    Trento & 1.1197 & 13    & 0.91099 & 12    & 1.18  & 0.92043 & 11    & 3.3   & 0.90233 & 12    & 1.28 \\
    Napoli II & 1.1156 & 14    & 0.913965 & 11    & 1.2   & 0.92258 & 10    & 3.34  & 0.90291 & 11    & 1.31 \\
    Milano Bicocca & 1.1024 & 15    & 0.88101 & 14    & 1.21  & 0.87323 & 13    & 3.43  & 0.90208 & 13    & 1.29 \\
    Marche & 1.0952 & 16    & 0.755735 & 18    & 1.19  & 0.73493 & 19    & 3.37  & 0.75044 & 20    & 1.28 \\
    Bologna & 1.0680 & 17    & 0.74515 & 19    & 1.31  & 0.76815 & 18    & 3.65  & 0.75085 & 19    & 1.42 \\
    Salento & 1.0618 & 18    & 0.66369 & 21    & 1.28  & 0.60666 & 21    & 3.64  & 0.68925 & 21    & 1.36 \\
    Milano & 1.0510 & 19    & 0.761015 & 17    & 1.34  & 0.81407 & 16    & 3.7   & 0.75799 & 18    & 1.46 \\
    Padova & 1.0483 & 20    & 0.729325 & 20    & 1.31  & 0.73434 & 20    & 3.67  & 0.77247 & 16    & 1.41 \\
    Ferrara & 1.0392 & 21    & 0.7651 & 16    & 1.34  & 0.77155 & 17    & 3.74  & 0.77105 & 17    & 1.44 \\
    Udine & 1.0301 & 22    & 0.616185 & 22    & 1.31  & 0.53283 & 22    & 3.77  & 0.59436 & 22    & 1.43 \\
    della Calabria & 1.0052 & 23    & 0.159515 & 30    & 1.4   & 0.12672 & 33    & 3.94  & 0.17688 & 30    & 1.51 \\
    Roma La Sapienza & 1.0047 & 24    & 0.13766 & 31    & 1.46  & 0.32741 & 28    & 3.92  & 0.12026 & 32    & 1.60 \\
    Milano Politecnico & 0.9991 & 25    & 0.06818 & 34    & 1.46  & 0.12891 & 32    & 3.98  & 0.06040 & 35    & 1.60 \\
    Piemonte Orientale & 0.9889 & 26    & 0.50371 & 23    & 1.38  & 0.43916 & 24    & 4.02  & 0.53354 & 24    & 1.45 \\
    Napoli Federico II & 0.9783 & 27    & 0.46423 & 25    & 1.42  & 0.34901 & 26    & 4.04  & 0.45277 & 25    & 1.55 \\
    Sannio & 0.9762 & 28    & 0.112775 & 32    & 1.58  & 0.15709 & 31    & 4.16  & 0.12287 & 31    & 1.73 \\
    Cagliari & 0.9733 & 29    & 0.480355 & 24    & 1.48  & 0.46175 & 23    & 4.16  & 0.56943 & 23    & 1.57 \\
    Firenze & 0.9644 & 30    & 0.404845 & 26    & 1.46  & 0.34262 & 27    & 4.1   & 0.43252 & 26    & 1.58 \\
    Parma & 0.9608 & 31    & 0.33844 & 27    & 1.54  & 0.35131 & 25    & 4.26  & 0.30857 & 28    & 1.69 \\
    Urbino Carlo Bo & 0.9444 & 32    & 0.218815 & 29    & 1.69  & 0.31418 & 29    & 4.39  & 0.22355 & 29    & 1.86 \\
    Salerno & 0.9373 & 33    & 0.24665 & 28    & 1.53  & 0.17767 & 30    & 4.37  & 0.31840 & 27    & 1.63 \\
    Torino & 0.9032 & 34    & 0.081775 & 33    & 1.63  & 0.06161 & 35    & 4.57  & 0.09362 & 33    & 1.77 \\
    L'Aquila & 0.8862 & 35    & 0.01643 & 45    & 1.72  & 0.01884 & 44    & 4.72  & 0.01172 & 45    & 1.91 \\
    Siena & 0.8804 & 36    & 0.06294 & 36    & 1.76  & 0.11009 & 34    & 4.72  & 0.04854 & 37    & 1.96 \\
    Perugia & 0.8682 & 37    & 0.028165 & 42    & 1.75  & 0.03985 & 38    & 4.75  & 0.03183 & 41    & 1.91 \\
    Modena e Reggio Emilia & 0.8644 & 38    & 0.06435 & 35    & 1.69  & 0.03189 & 40    & 4.81  & 0.06141 & 34    & 1.86 \\
    Bari  & 0.8514 & 39    & 0.010645 & 46    & 1.73  & 0.00500 & 46    & 4.87  & 0.01020 & 46    & 1.90 \\
    Camerino & 0.8380 & 40    & 0.04613 & 38    & 1.75  & 0.03013 & 41    & 4.97  & 0.04830 & 38    & 1.92 \\
    Chieti e Pescara & 0.8167 & 41    & 0.039 & 40    & 1.9   & 0.05355 & 36    & 5.1   & 0.04206 & 39    & 2.08 \\
    Catania & 0.8158 & 42    & 0.002995 & 47    & 1.88  & 0.00346 & 47    & 5.14  & 0.00370 & 47    & 2.05 \\
    Venezia Ca Foscari & 0.8056 & 43    & 0.0433 & 39    & 1.83  & 0.03373 & 39    & 5.15  & 0.02949 & 42    & 2.08 \\
    Basilicata & 0.7963 & 44    & 0.024345 & 43    & 1.86  & 0.01782 & 45    & 5.16  & 0.03771 & 40    & 2.00 \\
    Genova & 0.7937 & 45    & 4.00E-05 & 50    & 1.91  & 0.00002 & 50    & 5.27  & 0.00005 & 50    & 2.09 \\
    Reggio Calabria & 0.7899 & 46    & 0.055945 & 37    & 1.86  & 0.03995 & 37    & 5.24  & 0.05411 & 36    & 2.06 \\
    Messina & 0.7480 & 47    & 0.002355 & 48    & 1.95  & 0.00061 & 49    & 5.51  & 0.00201 & 48    & 2.15 \\
    Sassari & 0.7143 & 48    & 0.023 & 44    & 2.15  & 0.02680 & 43    & 5.73  & 0.02723 & 43    & 2.33 \\
    Napoli Parthenope & 0.7000 & 49    & 0.00131 & 49    & 2.06  & 0.00063 & 48    & 5.78  & 0.00140 & 49    & 2.28 \\
    Roma UNINETTUNO & 0.5952 & 50    & 0.029 & 41    & 2.42  & 0.02956 & 42    & 6.4   & 0.02176 & 44    & 2.71 \\
    \hline
    \end{tabular}%
    \vspace{0.2cm}
  \caption{Final rankings for the scientific area 01/A - Mathematics. We have denoted with (A) the choice $a_i=1, i=1, \ldots,5$, with (B) $a_1=a_2=a_3=3,a_4=1,a_5=0$, with (C) $a_1=a_2=a_4=1,a_3=1.5,a_5=0$. For each score calculated we report the associated rankings. Observations are ordered according to the VQR ranking based on $R$ values}
  \label{tab:01Arank}%
\end{table}%

\begin{table}[htbp]
  \centering
\tiny

\begin{tabular}{|l|cc|ccc|ccc|ccc|}
\hline
 & & & \multicolumn{3}{|c|}{(A)} & \multicolumn{3}{|c|}{(B)} & \multicolumn{3}{|c|}{(C)}\\\hline
  & &       &       &       &       &       &       &       &       &       &  \\
University  & $R$ & Rank & ${\mathcal S}_{\mathcal{C}}(U)$ & Rank   & $\delta(U)$ &${\mathcal S}_{\mathcal{C}}(U)$ &  Rank     & $\delta(U)$       & ${\mathcal S}_{\mathcal{C}}(U)$& Rank & $\delta(U)$ \\
      & &       &       &       &       &       &       &       &       &       &  \\
    \hline

    Milano Politecnico & 1.7647 & 1     & 0.99997 & 2     & 0.33  & 0.99999 & 3     & 1.00  & 0.99999 & 2     & 0.33 \\
    Ferrara & 1.6667 & 2     & 0.97731 & 13    & 0.50  & 0.97776 & 13    & 1.50  & 0.97776 & 13    & 0.50 \\
    Roma LUISS & 1.6667 & 2     & 0.99373 & 10    & 0.50  & 0.99380 & 9     & 1.50  & 0.99292 & 10    & 0.56 \\
    Milano & 1.6373 & 3     & 0.99990 & 4     & 0.55  & 0.99990 & 4     & 1.65  & 0.99993 & 4     & 0.55 \\
    Macerata & 1.6176 & 4     & 0.99792 & 8     & 0.58  & 0.99775 & 7     & 1.75  & 0.99821 & 8     & 0.58 \\
    Torino Politecnico & 1.4951 & 5     & 0.91082 & 18    & 0.88  & 0.92918 & 18    & 2.38  & 0.90854 & 18    & 0.94 \\
    Milano Bocconi & 1.4764 & 6     & 0.99997 & 2     & 0.94  & 1.00000 & 1     & 2.47  & 0.99996 & 3     & 1.03 \\
    Padova & 1.4537 & 7     & 1.00000 & 1     & 0.91  & 1.00000 & 1     & 2.59  & 1.00000 & 1     & 0.96 \\
    Sassari & 1.4461 & 8     & 0.99870 & 6     & 0.88  & 0.99833 & 6     & 2.63  & 0.99874 & 6     & 0.94 \\
    Perugia & 1.3313 & 9     & 0.99818 & 7     & 1.11  & 0.99761 & 8     & 3.21  & 0.99848 & 7     & 1.17 \\
    Urbino Carlo Bo & 1.3235 & 10    & 0.95960 & 15    & 1.08  & 0.94741 & 16    & 3.25  & 0.96497 & 15    & 1.12 \\
    Venezia Ca' Foscari & 1.3106 & 11    & 0.99974 & 5     & 1.11  & 0.99899 & 5     & 3.32  & 0.99980 & 5     & 1.16 \\
    Trento & 1.2717 & 12    & 0.99497 & 9     & 1.23  & 0.99342 & 10    & 3.51  & 0.99640 & 9     & 1.29 \\
    Piemonte Orientale & 1.2572 & 13    & 0.88657 & 19    & 1.35  & 0.92947 & 17    & 3.59  & 0.88162 & 19    & 1.47 \\
    Parma & 1.2567 & 14    & 0.96847 & 14    & 1.23  & 0.95430 & 15    & 3.59  & 0.97370 & 14    & 1.30 \\
    Chieti e Pescara & 1.2561 & 15    & 0.99164 & 11    & 1.19  & 0.99154 & 11    & 3.37  & 0.99207 & 11    & 1.27 \\
    Brescia & 1.2255 & 16    & 0.92867 & 17    & 1.25  & 0.87394 & 19    & 3.75  & 0.94654 & 16    & 1.29 \\
    Modena e Reggio Emilia & 1.1928 & 17    & 0.84874 & 20    & 1.42  & 0.87120 & 20    & 3.92  & 0.83820 & 20    & 1.54 \\
    Bologna & 1.1596 & 18    & 0.98210 & 12    & 1.46  & 0.98041 & 12    & 4.12  & 0.98151 & 12    & 1.60 \\
    Firenze & 1.1410 & 19    & 0.94837 & 16    & 1.51  & 0.95464 & 14    & 4.18  & 0.93483 & 17    & 1.66 \\
    Napoli II & 1.1111 & 20    & 0.66715 & 25    & 1.44  & 0.54547 & 27    & 4.33  & 0.67541 & 25    & 1.56 \\
    Torino & 1.1099 & 21    & 0.83949 & 21    & 1.55  & 0.83738 & 21    & 4.34  & 0.81922 & 21    & 1.70 \\
    Genova & 1.0873 & 22    & 0.63533 & 27    & 1.55  & 0.56815 & 26    & 4.45  & 0.61895 & 27    & 1.70 \\
    Milano Bicocca & 1.0565 & 23    & 0.77716 & 22    & 1.62  & 0.71746 & 23    & 4.61  & 0.79543 & 22    & 1.77 \\
    Bergamo & 1.0407 & 24    & 0.72058 & 24    & 1.62  & 0.60872 & 25    & 4.69  & 0.76130 & 23    & 1.73 \\
    Marche & 1.0392 & 25    & 0.73560 & 23    & 1.70  & 0.71927 & 22    & 4.70  & 0.73174 & 24    & 1.83 \\
    Salerno & 1.0114 & 26    & 0.42971 & 30    & 1.74  & 0.41217 & 31    & 4.84  & 0.48007 & 30    & 1.87 \\
    Roma Tor Vergata & 1.0076 & 27    & 0.66017 & 26    & 1.75  & 0.65762 & 24    & 4.86  & 0.63218 & 26    & 1.93 \\
    Udine & 1.0074 & 28    & 0.52801 & 28    & 1.69  & 0.43525 & 28    & 4.86  & 0.53518 & 28    & 1.84 \\
    Pisa  & 0.9741 & 29    & 0.33328 & 32    & 1.87  & 0.41446 & 29    & 5.03  & 0.30913 & 32    & 2.08 \\
    Pavia & 0.9617 & 30    & 0.50671 & 29    & 1.76  & 0.41242 & 30    & 5.10  & 0.50421 & 29    & 1.93 \\
    della Calabria & 0.9447 & 31    & 0.16967 & 37    & 1.86  & 0.15948 & 35    & 5.18  & 0.16538 & 36    & 2.06 \\
    Cagliari & 0.9276 & 32    & 0.35820 & 31    & 1.81  & 0.23802 & 32    & 5.27  & 0.39905 & 31    & 1.94 \\
    Palermo & 0.8554 & 33    & 0.12075 & 40    & 2.01  & 0.09586 & 38    & 5.64  & 0.13415 & 39    & 2.21 \\
    Milano Cattolica & 0.8507 & 34    & 0.05308 & 43    & 2.05  & 0.04405 & 42    & 5.66  & 0.06624 & 42    & 2.23 \\
    Napoli Federico II & 0.8287 & 35    & 0.02121 & 44    & 2.01  & 0.00506 & 47    & 5.77  & 0.02201 & 44    & 2.22 \\
    L'Aquila & 0.8088 & 36    & 0.20405 & 33    & 2.13  & 0.20193 & 33    & 5.88  & 0.19403 & 35    & 2.38 \\
    Insubria & 0.7994 & 37    & 0.20072 & 34    & 2.08  & 0.14105 & 36    & 5.92  & 0.20839 & 33    & 2.27 \\
    Cassino & 0.7843 & 38    & 0.15338 & 38    & 2.08  & 0.09103 & 40    & 6.00  & 0.14619 & 38    & 2.31 \\
    Salento & 0.7608 & 39    & 0.01263 & 45    & 2.28  & 0.01751 & 44    & 6.12  & 0.01532 & 45    & 2.50 \\
    Foggia & 0.7549 & 40    & 0.07633 & 41    & 2.15  & 0.04642 & 41    & 6.15  & 0.05961 & 43    & 2.43 \\
    Napoli Parthenope & 0.7376 & 41    & 0.17628 & 35    & 2.19  & 0.09225 & 39    & 6.24  & 0.20028 & 34    & 2.40 \\
    Roma LUMSA & 0.7190 & 42    & 0.17231 & 36    & 2.33  & 0.18438 & 34    & 6.33  & 0.15361 & 37    & 2.67 \\
    Roma La Sapienza & 0.6790 & 43    & 0.00000 & 48    & 2.46  & 0.00000 & 48    & 6.54  & 0.00000 & 48    & 2.75 \\
    Sannio & 0.6398 & 44    & 0.01121 & 46    & 2.42  & 0.00724 & 46    & 6.74  & 0.01463 & 46    & 2.66 \\
    Roma Europea & 0.5229 & 45    & 0.13016 & 39    & 2.67  & 0.12068 & 37    & 7.33  & 0.12100 & 40    & 3.00 \\
    Napoli Orientale & 0.4902 & 46    & 0.01111 & 47    & 2.83  & 0.01714 & 45    & 7.50  & 0.01295 & 47    & 3.17 \\
    Messina & 0.4256 & 47    & 0.00000 & 48    & 2.95  & 0.00000 & 48    & 7.83  & 0.00000 & 48    & 3.28 \\
    Teramo & 0.3676 & 48    & 0.05957 & 42    & 2.88  & 0.03435 & 43    & 8.13  & 0.06922 & 41    & 3.19 \\
    Bari  & 0.3650 & 49    & 0.00000 & 48    & 3.03  & 0.00000 & 48    & 8.14  & 0.00000 & 48    & 3.41 \\
       \hline
    \end{tabular}%
    \vspace{0.2cm}

  \caption{Final rankings for the scientific area 13/D - Statistics and Mathematical Methods for Decisions. We have denoted with (A) the choice $a_i=1, i=1, \ldots,5$, with (B) $a_1=a_2=a_3=3,a_4=1,a_5=0$, with (C) $a_1=a_2=a_4=1,a_3=1.5,a_5=0$.
  For each score calculated we report the associated rankings. Observations are ordered according to the VQR ranking based on $R$ values}
  \label{tab:13Drank}%
\end{table}%

\section{Conclusions and remarks}
As shown in Section \ref{VQR}, the results obtained applying the procedures developed in the present paper can be very different from the ranking obtained by ANVUR and used for funding Italian universities. The more one area is heterogeneous, either with respect to the numerousness either relatively to the specific research domains, the more the ANVUR  $R$ score becomes rough and the corresponding results differ from ours. For instance, we observe for the area of Mathematics (see Table \ref{tab:01Arank}) that the University of Pisa Normale loses three positions if evaluated according to geometric score instead of considering the ANVUR $R$ score. For the University of Bergamo and for  Milano Politecnico, the loss is more evident with the first one losing on average 6 positions and the second losing on average 8 positions, with a small variability depending on the $a_{i}$ constants chosen. On the other hand there are universities for which there is an evident improvement in the ranking, in some cases of even 10 positions, if the geometric score is used.

Similar considerations can be also done for the area of Statistics and Mathematical Methods for Decisions (see Table \ref{tab:13Drank}).

In fact, as pointed out in Section \ref{VQR} one can not in principle compare two aggregates of different size without risking having a funnel effect. In order to overcome this issue, in some areas  ANVUR divided universities in three classes according to the number of researchers of each aggregate (big, medium and small) and, for such areas, only the ranking within these three dimensional classes was given (but in any case the computation of the $R$ score, used for funding allocations, was made regardless the dimensional class to which each university was belonging).  However, while this remedy could mitigate some perverse effect, it can not prevent the appareance of funnel effects in each dimensional class. Moreover, in this way the analysis is inevitably less informative, since each aggregate is compared with a fewer number of other aggregates. We stress that to calculate the geometric score there is no need to make distinctions on the base of universities' dimensions as, according to the definition of  ${\mathcal S}_{\mathcal{C}}(U)$, each university $U$ is compared to all ideal universities having the same size as $U$.

 On the other hand, the scientific homogeneity is another important feature which is not considered in the ANVUR analysis but it is incapsulated in our geometric score. The best example in this way is provided by the Mathematics area. Here we have very different outcomes for each scientific discipline, reflecting the diverse publication customs and trends for the various areas of Mathematics, as well as the different methods of assessment which were used (the MCQ score for Pure Mathematics, Impact Factor for Applied Mathematics, peer review for History of Mathematics). From  Table \ref{VQR-math} it is clear that departments with  higher number of professors dealing with History of Mathematics are unfairly penalized by the ANVUR's analysis methods. With this respect, we observe that the three universities of Pisa Normale, Bergamo and Milano Politecnico, that we have mentioned above among the ones rewarded from the VQR ranking, did not count any professor or reasearcher in that scientific sector. On the contrary  departments with  higher number of professors dealing with Applied Mathematics obtain, on average, better results.

\begin{table}[htbp]\label{VQR-math}
  \centering
\tiny
  \caption{Number of total expected products and their repartition in the 5  classes of VQR 2011-2014, for each scientific sector of the area 01/A - Mathematics. The proportion of products in the 5 classes is also reported}
    \begin{tabular}{|c|c|ccccc|}
    \hline
    {scientific sector}  & {expected products} & {A}  & {B} & {C} & {D} & {E} \\
    \hline
    MAT/01 & 72 & 25	& 20	&7	&12&	8      \\
                &  & 0.347   & 0.278 & 0.097 & 0.167 & 0.167      \\
MAT/02 & 319 &  67	& 97&	59&	32&	64        \\
            &       & 0.21   & 0.304 & 0.185 & 0.10 & 0.345        \\
MAT/03 & 800 & 246	& 190 &	106 &	82 &	176    \\
            &       & 0.308   & 0.238 & 0.133 & 0.103 & 0.385    \\
MAT/04 & 132 & 18	& 45 &	28 &	20 &	21      \\
            &       & 0.136   & 0.311 & 0.212 & 0.152 & 0.227       \\
MAT/05 & 1545 & 616&	388&	220&	94	& 227   \\
            &          & 0.399   & 0.251 & 0.142 & 0.061 & 0.27   \\
MAT/06 & 255 &97	& 71 &	42 &	22	 & 23    \\
            &       & 0.38   & 0.278 & 0.165 & 0.086 & 0.169     \\
MAT/07 & 609 & 197	& 147 &	98	& 88	 & 79      \\
            &       & 0.324   & 0.241 & 0.161 & 0.145 & 0.22      \\
MAT/08 & 563 & 245	& 143 &	83 &	36 &	56       \\
            &        & 0.435   & 0.254 & 0.147 & 0.064 & 0.17       \\
MAT/09 & 230 &169	& 68 &	30 &	14	&16       \\
            &        & 0.569   & 0.229 & 0.101 & 0.047 & 0.082       \\
    \hline
    \end{tabular}
\end{table}

\begin{table}[htbp]\label{VQR-stat}
  \centering
\tiny
  \caption{Number of total expected products and their repartition in the 5  classes of VQR 2011-2014, for each scientific sector of the area 13/D - Statistics and Mathematical Methods for Decisions. The proportion of products in the 5 classes is also reported}
    \begin{tabular}{|c|c|ccccc|}
    \hline
    {scientific sector}  & {expected products} & {A}  & {B} & {C} & {D} & {E}  \\
    \hline
   SECS-S/01 & 794 &    241&	209&	94	&91	&109      \\
                   &        &0.325   & 0.281 &0.126 & 0.122 & 0.146\\
   SECS-S/02 & 45 &  11 &	10 &	5	&5&	11       \\
               &        & 0.267 & 0.244 & 0.111 & 0.111 & 0.267  \\
   SECS-S/03 & 281 & 44	& 61 &	41 &	45 &	67          \\
                    &      & 0.171 & 0.235 & 0.160 & 0.174 & 0.260          \\
   SECS-S/04 & 131 & 28	& 24	&14	&29&	27      \\
                   &        & 0.229  & 0.198  & 0.115 & 0.237 & 0.221      \\
   SECS-S/05 & 130 & 20	& 22 &	33 &	27	 & 21     \\
                    &      & 0.162 & 0.177 & 0.269 & 0.223 & 0.169     \\
   SECS-S/06 & 776 & 152 &	208 &	85 &	89	& 142        \\
                    &      & 0.224 & 0.308 & 0.126 & 0.132 & 0.21        \\
    \hline
    \end{tabular}
\end{table}

Somehow ANVUR itself was aware of the aforementioned limits of its methods, so that for the program ``Departments of Excellence'' a different methodology, conceptually much similar to ours, was  introduced (see \cite{poggi} and \cite{ANVUR:2017-1}). Without entering into details, we point out that such methodology, which uses the Central Limit Theorem, assumes  the independence of the assessments received by each publication. Such independence assumption, however, is  unrealistic, especially for smaller sectors, as well as for areas with large numbers of coauthors.

%
%

\appendix

\section{The proof of \eqref{eq:delta-n}}

In this appendix we shall give a formal proof of \eqref{eq:delta-n}. Let's start with the precise definition of $\delta(P_0)=\delta(P_0,P_1)$. For this, let $\Delta_n$ be the $n$-dimensional simplex with vertices $P_{1}=(1,0,\ldots,0), P_{2}=(0,1,\ldots,0), \ldots, P_{n+1}=(0,0,\ldots,1)\in\R^{n+1}$ and let
$P_0=(x_1^0, \ldots, x_{n+1}^0)\in\Delta_n$ be a point of the simplex, so that we have $\sum_{i=1}^{n+1} x_i^0=1$.
Let $\pi^1$ be the hyperplane of $\R^{n+1}$ through $P_0$ and parallel to the plane containing the vertices $\{P_2,\ldots,P_{n+1}\}$ and let $\{P_1^1,\ldots,P_n^1\}$ be the points defined by $P_i^1=\pi^1\cap\ell_{1,i+1}^1, i=1,\ldots,n,$ where $\ell_{1,j}^1$ is the line trough $P_1$ and $P_j$, $j=2,\ldots,n+1$.
Next, consider the hyperplane $\pi^2$ trough $P_0$ and parallel to the plane containing the vertices $P_2^1,\ldots,P_n^1$ and, as before, define the points
$\{P_1^2,\ldots,P_{n-1}^2\}$ by  $P_i^2=\pi^2\cap\ell_{1,i+1}^2, i=1,\ldots,n-1,$  where $\ell_{1,j}^2$ is the line trough $P_1^1$ and $P_j^1$, $j=2,\ldots,n$. Applying this procedure fot $n-1$ steps we obtain a sequence of points $P_1,P_1^1,\ldots, P_1^{n-1}$.
Now for a fixed set of positive real constants $a=\{a_1,\ldots,a_{n+1}\}$ we define
\begin{align}\label{eq-def-delta}
\delta(P_0)&=\delta(P_0,P_1)\\
&=\left\| P_1^{n-1} -   P_0\right\|_{g}+\left\| P_1^{n-2} -  P_1^{n-1}\right\|_{g}+\left\|P_1^{n-3}-P_1^{n-2}\right\|_{g}+\cdots+\left\|P_1-P_1^{1}\right\|_{g}.\nonumber
\end{align}
Here $\|\cdot\|_{g}$ represents the norm  with respect to the unique inner product $\langle ,\rangle_{g}$ in ${\mathbb R}^{n+1}$ defined by
$$
\langle \mathbf{v}_{i}, \mathbf{v}_{j}\rangle_{g}=\delta_{i,j}\,
a_i a_j
$$
where, denoting by $O$ the origin of ${\mathbb R}^{n+1}$,
$$
\mathbf{v}_{1}=P_1-O\,,\quad \mathbf{v}_{2}=P_2-P_1\,,\ldots, \mathbf{v}_{n+1}=P_{n+1}-P_n\,.
$$
Note that, as the origin does not belong to the simplex and $P_1,\ldots, P_{n+1}$ are affinely independent,  $\{\mathbf{v}_{1},\ldots, \mathbf{v}_{n+1}\}$ forms a  basis of ${\mathbb R}^{n+1}$.
\begin{theorem}
Let $\Delta_n$ be the canonical $n$-dimensional simplex  and let $P_0=(x_1^0, \ldots, x_{n+1}^0)\in\Delta_n$ be a point of the simplex. Let $a=\{a_1,\ldots,a_{n+1}\}$ be a set of positive real constants. Then, with the above notation, we have
\begin{eqnarray}\label{eq-main-delta}
\delta(P_0)&=&a_{1}+\cdots + a_{n+1} - (a_{1}+\cdots + a_{n+1})x_{1}^{0} - (a_{2}+\cdots + a_{n+1}) x_{2}^{0} -\cdots\nonumber\\
&&  - (a_{n}+a_{n+1}) x_{n}^{0} - a_{n+1} x_{n+1}^{0}
\end{eqnarray}
\end{theorem}
\begin{proof}
We shall prove the theorem by induction on the dimension $n$ of the simplex. If $n=1$, the simplex degenerates into the segment from the point $P_{1}=(1,0)$ to $P_{2}=(0,1)$. Given a point $P_0=(x_1^0, x_2^0)=(x^0_{1},1-x^0_{1})\in\Delta_1$, from the definition \eqref{eq-def-delta} of $\delta(P_0)$ we have
$\delta(P_0)=d_g(P_0,P_1)=a_1 (1-x^0_{1})$ which coincides with \eqref{eq-main-delta} for $n=1$.

Let now assume that the \eqref{eq-main-delta}  is true for any canonical simplex of dimension $n$ and for any set $a$ of positive real constants with $|a|=n+1$. Let prove that it is valid for the canonical $(n+1)$-dimensional simplex and for any set of $a=\{a_1,\ldots, a_{n+2}\}$ of positive real constants. Let $\Delta_{n+1}$ be the $(n+1)$-dimensional simplex with vertices $P_{1}=(1,0,\ldots,0), P_{2}=(0,1,\ldots,0), \ldots, P_{n+2}=(0,0,\ldots,1)\in\R^{n+2}$ and let $P_0=(x_1^0, \ldots, x_{n+2}^0)\in\Delta_{n+1}$ be a point of the simplex. Let's first compute the coordinates of the points $P_i^1, i=1,\ldots,n+1$. A straightforward computation shows that the hyperplane $\pi^1$ has Cartesian equation
$$
\langle N^1, P-P_0\rangle=0
$$
where $N^1=(n+1,-1,-1,\ldots,-1)$ is the normal vector to $\pi^1$ while $P=(x_1,\ldots, x_{n+2})$. The line $\ell_{1,j}^1$ is parametrised by
$$
P=P_1+t (P_j-P_1).
$$
Substituting the latter in the equation of the plane $\pi^1$ we obtain that the intersection point correspond to the parameter $t=1-x_1^0$. We thus have
$$
P_i^1=P_1+(1-x_1^0)(P_{i+1}-P_1)=x_1^0 P_1+(1-x_1^0)P_{i+1},\quad i=1,\ldots,n+1.
$$
In coordinates we have
$$
\begin{cases}
P_1^1=(x_1^0,1-x_1^0,0,0,\ldots,0)\\
P_2^1=(x_1^0,0,1-x_1^0,0,\ldots,0)\\
\quad \vdots\\
P_{n+1}^1=(x_1^0,0,\ldots,0, 1-x_1^0)\\
\end{cases}
$$
and
$$
P_0=(x_1^0, \ldots, x_{n+2}^0).
$$
The points $P_1^1,\ldots, P_{n+1}^1$ and $P_0$  lie in the hyperplane $x_1=x_1^0$ and we can identify them with the points
$$
\begin{cases}
P_1^1=(0,1-x_1^0,0,0,\ldots,0)\\
P_2^1=(0,0,1-x_1^0,0,\ldots,0)\\
\quad \vdots\\
P_{n+1}^1=(0,0,\ldots,0, 1-x_1^0)\\
P_0=(0,x_2^0, \ldots, x_{n+2}^0)
\end{cases}
$$
via the isometry between the hyperplane $x_1=x_1^0$ and the hyperplane $x_1=0$. If we now apply the homothety
\begin{equation}\label{eq-omothety}
P\mapsto \frac{1}{1-x_1^0}P
\end{equation}
the points $P_1^1,\ldots, P_{n+1}^1$ and $P_0$ are mapped to
$$
\begin{cases}
\bar{P}_1^1=(0,1,0,0,\ldots,0)\\
\bar{P}_2^1=(0,0,1,0,\ldots,0)\\
\quad \vdots\\
\bar{P}_{n+1}^1=(0,0,\ldots,0, 1)\\
\bar{P}_0=\left(0,\frac{x_2^0}{1-x_1^0}, \ldots, \frac{x_{n+2}^0}{1-x_1^0}\right),
\end{cases}
$$
which can be identified with the following points in $\R^{n+1}$
$$
\begin{cases}
\bar{P}_1^1=(1,0,0,\ldots,0)\\
\bar{P}_2^1=(0,1,0,\ldots,0)\\
\quad \vdots\\
\bar{P}_{n+1}^1=(0,\ldots,0, 1)\\
\bar{P}_0=\left(\frac{x_2^0}{1-x_1^0}, \ldots, \frac{x_{n+2}^0}{1-x_1^0}\right).
\end{cases}
$$
At this stage, we can apply the inductive hypothesis to the canonical $n$-dimensional simplex with vertices $\bar{P}_1^1,\ldots,\bar{P}_{n+1}^1$, with respect to the set of real constants $\bar{a}=\{a_2,\ldots, a_{n+2}\}$, and considering $\bar{P}_0=\left(\frac{x_2^0}{1-x_1^0}, \ldots, \frac{x_{n+2}^0}{1-x_1^0}\right)$ as a point of the simplex. We have
\[
\begin{split}
\delta(\bar{P}_0,\bar{P}_1^1)=&\frac{1}{1-x_1^0}\Big((a_{2}+\cdots + a_{n+2})(1- x_1^0)- (a_{2}+\cdots + a_{n+2})x_{2}^{0} - \cdots \\
&- (a_{n+1}+a_{n+2}) x_{n+1}^{0} - a_{n+2} x_{n+2}^{0}\Big).
\end{split}
\]
Taking into account that the homothety \eqref{eq-omothety} change the norm $\|\;\|_g$ by the factor $1/(1-x_1^0)$, we find that
\[
\begin{split}
\delta(P_0,P_1^1)=(1-x_1^0)\, \delta(\bar{P}_0,\bar{P}_1^1)=&
(a_{2}+\cdots + a_{n+2})(1- x_1^0)- (a_{2}+\cdots + a_{n+2})x_{2}^{0} - \cdots \\
&- (a_{n+1}+a_{n+2}) x_{n+1}^{0} - a_{n+2} x_{n+2}^{0}
\end{split}
\]
Finally
\begin{eqnarray*}
\delta(P_0,P_1)&=&\delta(P_0,P_1^1)+\left\| P_1 -   P_1^1\right\|_{g}\\
&=& (a_{2}+\cdots + a_{n+2})(1- x_1^0)- (a_{2}+\cdots + a_{n+2})x_{2}^{0} - \cdots \\
&&- (a_{n+1}+a_{n+2}) x_{n+1}^{0} - a_{n+2} x_{n+2}^{0}+(1- x_1^0)a_1\\
&=&
(a_1+a_{2}+\cdots + a_{n+2})-(a_1+a_{2}+\cdots + a_{n+2}) x_1^0- (a_{2}+\cdots + a_{n+2})x_{2}^{0} - \cdots \\
&&- (a_{n+1}+a_{n+2}) x_{n+1}^{0} - a_{n+2} x_{n+2}^{0}
\end{eqnarray*}
which is exactly \eqref{eq-main-delta} for $n+1$.
\end{proof}


\begin{thebibliography}{99}



\bibitem{Abramo:2016} G. Abramo, C.A. D'Angelo, \emph{Refrain from adopting the combination of citation and journal metrics to grade publications, as used in the Italian national research assessment exercise (VQR 2011–2014)}, Scientometrics \textbf{109} (2016), 2053--2065

\bibitem{ANVUR:2017-1} ANVUR, \emph{Nota metodologica sul calcolo dell’indicatore ISPD}. Techncal report available at \url{https://www.anvur.it/wp-content/uploads/2018/04/Nota_metodologica_ISPD_AN_.pdf} (2017)

\bibitem{ANVUR:2017-2} ANVUR, \emph{Final area report Group of Evaluation Experts}. Technical report available at \url{https://www.anvur.it/rapporto-2016/} (2017)


\bibitem{Baccini} A. Baccini, G. De Nicolao, \emph{Do they agree? Bibliometric evaluation versus informed peer review in the Italian research assessment exercise}, Scientometrics \textbf{108} (2016), 1651--1671


\bibitem{DORA} R. Cagan, \emph{The San Francisco Declaration on Research Assessment}, Disease Models \& Mechanisms
(2013)  \textbf{6}, Editorial


\bibitem{Demetrescu} C. Demetrescu, F. Lupia, A. Mendicelli, A. Ribichini, F. Scarcello, M. Schaerf, \emph{On the Shapley value and its application to the Italian VQR research assessment exercise}, Journal of Informetrics \textbf{13} (2019), 87--104


\bibitem{Franceschini} F. Franceschini, D. Maisano, \emph{Critical remarks on the Italian research assessment exercise VQR 2011–2014}, Journal of Informetrics
\textbf{11} (2017), 337--357




\bibitem{Hoeffding: 1963} W. Hoeffding  \emph{Probability inequalities for sums of bounded random variables}  Journal of the American Statistical Association \textbf{58 (301)}, 13--30



\bibitem{poggi} G. Poggi, C. A. Nappi, \emph{Il Voto standardizzato per l’esercizio VQR 2004-2010}, Rassegna italiana di valutazione, a. XVIII, n. 59 (2014), 34--58

\bibitem{Spiegelhalter} D. J. Spiegelhalter, \emph{Funnel plots for comparing institutional performance}, Statistics in Medicine,  \textbf{24}, 8 (2005), 1185--1202

\bibitem{taylor1995} J. Taylor, \emph{A statistical analysis of the 1992 Research Assessment Exercise}, Journal of the Royal Statistical Society Series A - Statistics in Society \textbf{158} (1995), 241--261

\bibitem{varin2016} C. Varin, M. Cattelan, D. Firth, \emph{Statistical modelling of citation exchange between statistics journals}, Journal of the Royal Statistical Society Series A - Statistics in Society \textbf{179} (2016), 1--63

\bibitem{Weiss2019} C.H. Weiss, \emph{On some measures of ordinal variation}, Journal of Applied Statistics \textbf{46, (16)} (2019), 2905--2926

\bibitem{Weiss2020} C.H. Weiss, \emph{Distance-Based Analysis of Ordinal Data and Ordinal Time Series}, Journal of the American Statistical Association \textbf{ 531, (115)} (2020), 1189--1200










\end{thebibliography}
\end{document}